\documentclass[11pt]{article}

\def\showauthornotes{0}
\def\showtableofcontents{1}
\def\showkeys{0}
\def\showdraftbox{0}
\def\showcolorlinks{1}
\def\usemicrotype{1}
\def\showfixme{0}
\def\showdefine{0}

\usepackage{etex}

\usepackage[l2tabu, orthodox]{nag}

\usepackage{xspace,enumerate}

\usepackage[dvipsnames]{xcolor}

\usepackage[T1]{fontenc}
\usepackage[full]{textcomp}

\usepackage[american]{babel}

\usepackage{mathtools}

\usepackage{stmaryrd}

\usepackage{amsthm}

\newtheorem{theorem}{Theorem}[section]
\newtheorem*{theorem*}{Theorem}

\newtheorem*{proposition*}{Proposition}
\newtheorem{lemma}[theorem]{Lemma}
\newtheorem*{lemma*}{Lemma}
\newtheorem{corollary}[theorem]{Corollary}
\newtheorem*{conjecture*}{Conjecture}

\newtheorem*{fact*}{Fact}

\newtheorem*{hypothesis*}{Hypothesis}

\theoremstyle{definition}
\newtheorem{definition}[theorem]{Definition}

\theoremstyle{remark}
\newtheorem{claim}[theorem]{Claim}
\newtheorem*{claim*}{Claim}

\newtheorem*{remark*}{Remark}

\newtheorem*{observation*}{Observation}

\usepackage[letterpaper,
top=0.8in,
bottom=0.8in,
left=1.3in,
right=1.3in]{geometry}

\usepackage[varg]{pxfonts} %

\ifnum\showkeys=1
\usepackage[color]{showkeys}
\fi

\ifnum\showcolorlinks=1
\usepackage[
pagebackref,
letterpaper=true,
colorlinks=true,
urlcolor=blue,
linkcolor=blue,
citecolor=OliveGreen,
]{hyperref}
\fi

\ifnum\showcolorlinks=0
\usepackage[
pagebackref,
letterpaper=true,
colorlinks=false,
pdfborder={0 0 0}
]{hyperref}
\fi

\usepackage{prettyref}

\newcommand{\savehyperref}[2]{\texorpdfstring{\hyperref[#1]{#2}}{#2}}

\newrefformat{eq}{\savehyperref{#1}{\textup{(\ref*{#1})}}}
\newrefformat{lem}{\savehyperref{#1}{Lemma~\ref*{#1}}}
\newrefformat{def}{\savehyperref{#1}{Definition~\ref*{#1}}}
\newrefformat{thm}{\savehyperref{#1}{Theorem~\ref*{#1}}}
\newrefformat{cor}{\savehyperref{#1}{Corollary~\ref*{#1}}}
\newrefformat{cha}{\savehyperref{#1}{Chapter~\ref*{#1}}}
\newrefformat{sec}{\savehyperref{#1}{Section~\ref*{#1}}}
\newrefformat{app}{\savehyperref{#1}{Appendix~\ref*{#1}}}
\newrefformat{tab}{\savehyperref{#1}{Table~\ref*{#1}}}
\newrefformat{fig}{\savehyperref{#1}{Figure~\ref*{#1}}}
\newrefformat{hyp}{\savehyperref{#1}{Hypothesis~\ref*{#1}}}
\newrefformat{alg}{\savehyperref{#1}{Algorithm~\ref*{#1}}}
\newrefformat{rem}{\savehyperref{#1}{Remark~\ref*{#1}}}
\newrefformat{item}{\savehyperref{#1}{Item~\ref*{#1}}}
\newrefformat{step}{\savehyperref{#1}{step~\ref*{#1}}}
\newrefformat{conj}{\savehyperref{#1}{Conjecture~\ref*{#1}}}
\newrefformat{fact}{\savehyperref{#1}{Fact~\ref*{#1}}}
\newrefformat{prop}{\savehyperref{#1}{Proposition~\ref*{#1}}}
\newrefformat{prob}{\savehyperref{#1}{Problem~\ref*{#1}}}
\newrefformat{claim}{\savehyperref{#1}{Claim~\ref*{#1}}}
\newrefformat{relax}{\savehyperref{#1}{Relaxation~\ref*{#1}}}
\newrefformat{red}{\savehyperref{#1}{Reduction~\ref*{#1}}}
\newrefformat{part}{\savehyperref{#1}{Part~\ref*{#1}}}

\newcommand{\Sref}[1]{\hyperref[#1]{\S\ref*{#1}}}

\usepackage{nicefrac}

\let\nfrac=\nicefrac
\let\ffrac=\flatfrac

\newcommand{\half}{\nicefrac12}

\ifnum\usemicrotype=1
\usepackage{microtype}
\fi

\ifnum\showauthornotes=1
\newcommand{\Authornote}[2]{{\sffamily\small\color{red}{[#1: #2]}}}
\newcommand{\Authornotecolored}[3]{{\sffamily\small\color{#1}{[#2: #3]}}}
\newcommand{\Authorcomment}[2]{{\sffamily\small\color{gray}{[#1: #2]}}}
\newcommand{\Authorstartcomment}[1]{\sffamily\small\color{gray}[#1: }

\newcommand{\Authorfnote}[2]{\footnote{\color{red}{#1: #2}}}
\newcommand{\Authorfixme}[1]{\Authornote{#1}{\textbf{??}}}
\newcommand{\Authormarginmark}[1]{\marginpar{\textcolor{red}{\fbox{\Large #1:!}}}}
\else
\newcommand{\Authornote}[2]{}
\newcommand{\Authornotecolored}[3]{}
\newcommand{\Authorcomment}[2]{}
\newcommand{\Authorstartcomment}[1]{}

\newcommand{\Authorfnote}[2]{}
\newcommand{\Authorfixme}[1]{}
\newcommand{\Authormarginmark}[1]{}
\fi

\newcommand{\Dnote}{\Authornote{D}}
\newcommand{\Inote}{\Authornotecolored{ForestGreen}{I}}
\newcommand{\inote}{\Inote}

\newcommand{\Icomment}{\Authorcomment{I}}
\newcommand{\Dcomment}{\Authorcomment{D}}

\ifnum\showfixme=0

\fi

\usepackage{boxedminipage}

\newcommand{\paren}[1]{(#1)}
\newcommand{\Paren}[1]{\left(#1\right)}
\newcommand{\bigparen}[1]{\big(#1\big)}
\newcommand{\Bigparen}[1]{\Big(#1\Big)}
\newcommand{\abs}[1]{\lvert#1\rvert}
\newcommand{\Abs}[1]{\left\lvert#1\right\rvert}

\newcommand{\card}[1]{\lvert#1\rvert}

\newcommand{\set}[1]{\{#1\}}
\newcommand{\Set}[1]{\left\{#1\right\}}

\newcommand{\norm}[1]{\lVert#1\rVert}

\newcommand{\Bignorm}[1]{\Big\lVert#1\Big\rVert}
\newcommand{\snorm}[1]{\norm{#1}^2}

\newcommand{\Bigsnorm}[1]{\Bignorm{#1}^2}
\newcommand{\normo}[1]{\norm{#1}_1}

\newcommand{\iprod}[1]{\langle#1\rangle}

\newcommand{\Esymb}{\mathbb{E}}
\newcommand{\Psymb}{\mathbb{P}}

\DeclareMathOperator*{\E}{\Esymb}

\DeclareMathOperator*{\ProbOp}{\Psymb}

\renewcommand{\Pr}{\ProbOp}

\newcommand{\Prob}[2][]{\Pr_{{#1}}\Set{#2}}

 \usepackage{dsfont}
\usepackage{mathrsfs}

\newcommand{\textparen}[1]{\text{(#1)}}
\newcommand{\using}[1]{\textparen{using #1}}

\ifx\because\undefined
\newcommand{\because}[1]{\textparen{because #1}}
\else
\renewcommand{\because}[1]{\textparen{because #1}}
\fi
\newcommand{\by}[1]{\textparen{by #1}}

\newcommand{\bits}{\{0,1\}}

\newcommand{\vbig}{\vphantom{\bigoplus}}

\newcommand{\defeq}{\stackrel{\mathrm{def}}=}

\newcommand{\seteq}{\mathrel{\mathop:}=}

\newcommand{\from}{\colon}

\newcommand{\mper}{\,.}
\newcommand{\mcom}{\,,}

\newcommand\bdot\bullet

\newcommand{\Ind}{\mathbb I}

\DeclareMathOperator{\val}{val}

\DeclareMathOperator{\poly}{poly}

\newcommand{\Lovasz}{Lov\'asz\xspace}

\newcommand{\N}{\mathbb N}
\newcommand{\R}{\mathbb R}

\newcommand{\Rnn}{\R_+}

\newcommand{\problemmacro}[1]{\texorpdfstring{\textsc{#1}}{#1}\xspace}

\newcommand{\maxcut}{\problemmacro{max cut}}

\newcommand{\labelcover}{\problemmacro{label cover}}
\newcommand{\setcover}{\problemmacro{set cover}}

\renewcommand{\le}{\leqslant}
\renewcommand{\geq}{\geqslant}
\renewcommand{\ge}{\geqslant}

\ifnum\showdraftbox=1
\newcommand{\draftbox}{\begin{center}
  \fbox{%
    \begin{minipage}{2in}%
      \begin{center}%
          \Large\textsc{Working Draft}\\%
        Please do not distribute%
      \end{center}%
    \end{minipage}%
  }%
\end{center}
\vspace{0.2cm}}
\else
\newcommand{\draftbox}{}
\fi

\let\epsilon=\varepsilon

\numberwithin{equation}{section}

\newcommand{\MYstore}[2]{%
  \global\expandafter \def \csname MYMEMORY #1 \endcsname{#2}%
}

\newcommand{\MYload}[1]{%
  \csname MYMEMORY #1 \endcsname%
}

\newcommand{\MYnewlabel}[1]{%
  \newcommand\MYcurrentlabel{#1}%
  \MYoldlabel{#1}%
}

\newcommand{\MYdummylabel}[1]{}

\newcommand{\torestate}[1]{%
  \let\MYoldlabel\label%
  \let\label\MYnewlabel%
  #1%
  \MYstore{\MYcurrentlabel}{#1}%
  \let\label\MYoldlabel%
}

\newcommand{\restatetheorem}[1]{%
  \let\MYoldlabel\label
  \let\label\MYdummylabel
  \begin{theorem*}[Restatement of \prettyref{#1}]
    \MYload{#1}
  \end{theorem*}
  \let\label\MYoldlabel
}

\newcommand{\restatelemma}[1]{%
  \let\MYoldlabel\label
  \let\label\MYdummylabel
  \begin{lemma*}[Restatement of \prettyref{#1}]
    \MYload{#1}
  \end{lemma*}
  \let\label\MYoldlabel
}

\newcommand{\restateprop}[1]{%
  \let\MYoldlabel\label
  \let\label\MYdummylabel
  \begin{proposition*}[Restatement of \prettyref{#1}]
    \MYload{#1}
  \end{proposition*}
  \let\label\MYoldlabel
}

\newcommand{\restatefact}[1]{%
  \let\MYoldlabel\label
  \let\label\MYdummylabel
  \begin{fact*}[Restatement of \prettyref{#1}]
    \MYload{#1}
  \end{fact*}
  \let\label\MYoldlabel
}

\newcommand{\restate}[1]{%
  \let\MYoldlabel\label
  \let\label\MYdummylabel
  \MYload{#1}
  \let\label\MYoldlabel
}

\newcommand{\addreferencesection}{
  \phantomsection
  \addcontentsline{toc}{section}{References}
}

\newcommand{\sse}{\subseteq}

\newcommand{\e}{\epsilon}
\newcommand{\eps}{\epsilon}

\let\origparagraph\paragraph
\renewcommand{\paragraph}[1]{\origparagraph{#1.}}

\allowdisplaybreaks

\newcommand{\Id}{\mathrm{Id}}

\usepackage{tikz}
\usetikzlibrary{matrix}

\let\pref=\prettyref

\ifnum\showdefine=1
\newcommand{\define}[1]{\marginpar{\tiny \raggedright\color{gray}{\emph{#1}}}\emph{#1}}
\else
\newcommand{\define}[1]{\emph{#1}}
\fi

\renewcommand{\Ind}{\mathds 1}

\newcommand{\ot}{\otimes}

\let\Rnn\varvarRnn
\newcommand{\Assign}{\mathtt{Assign}}

\newcommand{\ok}{{\otimes k}}
\newcommand{\rval}{{\mathrm{val}_+}}
\newcommand{\fakerval}{{\lambda_+}}

\newcommand\remove[1]{}

\title{\bfseries Analytical Approach to Parallel Repetition}

\author{%
  Irit Dinur\thanks{%
    Department of Computer Science and Applied
    Mathematics, Weizmann Institute. Part of this work was done at
    Microsoft Research New England and Radcliffe Institute for
    Advanced Study. Work supported in part by an ISF grant number 1179/09 and by an ERC-StG grant number 239985.}
  \and David Steurer\thanks{%
    Computer Science
    Department, Cornell University. Part of this work was done at
    Microsoft Research New England.
    Research supported in part by NSF and the Alred P. Sloan Foundation.
  } }

\begin{document}

\maketitle
\draftbox
\thispagestyle{empty}

\begin{abstract}

We propose an analytical framework for studying parallel repetition, a basic product operation for one-round two-player games.
In this framework, we consider a relaxation of the value of projection games.
We show that this relaxation is multiplicative with respect to parallel repetition and that it provides a good approximation to the game value.
Based on this relaxation, we prove the following improved parallel repetition bound: For every projection game~$G$ with value at most~$\rho$, the $k$-fold parallel repetition $G^{\otimes k}$ has value at most
\begin{displaymath}
  \mathrm{val}(G^{\otimes k}) \le \left(\frac{2\sqrt{\rho}}{1+\rho} \right)^{k/2}.
\end{displaymath}
This statement implies a parallel repetition bound for projection games with low value~$\rho$.
Previously, it was not known whether parallel repetition decreases the value of such games.
This result allows us to show that approximating \textsc{set cover} to within factor $(1-\e)\ln n$ is NP-hard for every $\e>0$, strengthening Feige's quasi-NP-hardness and also building on previous work by Moshkovitz and Raz.

In this framework, we also show improved bounds for few parallel repetitions of projection games, showing that Raz's counterexample to strong parallel repetition is tight even for a small number of repetitions.

Finally, we also give a short proof for the NP-hardness of \textsc{label cover}$(1,\delta)$ for all~$\delta>0$, starting from the basic PCP theorem.

\end{abstract}


\clearpage

\ifnum\showtableofcontents=1
{
\tableofcontents
\thispagestyle{empty}
 }
\fi

\clearpage

\section{Introduction}
A \define{one-round two-player game} $G$ is specified
by a bipartite graph with vertex sets $U$ and $V$ and edges decorated by constraints $\pi\sse \Sigma \times \Sigma$ for an alphabet $\Sigma$.
The \emph{value} of a game is the maximum, over all assignments $f\from U\to\Sigma$ and $g\from
V\to\Sigma$, of the fraction of constraints satisfied (where
a constraint $\pi$ is satisfied if
$\Paren{f(u),g(v)}\in\pi$)
\begin{displaymath}
\val(G) = \max_{f,g}\Prob[u,v,\pi]{\vbig \bigparen{f(u),g(v)}\in \pi}.
\end{displaymath}
The term one-round two-player game stems from the following
scenario:
A referee interacts with two players, Alice and Bob.
Alice has a strategy $f\from U\to\Sigma$,
and Bob a strategy $g\from V\to\Sigma$.
A referee selects a random edge $u,v$ in $E$ and sends  $u$ as a question to
Alice, and $v$ as a question to Bob. Alice responds with $f(u)$ and Bob
with $g(v)$. They succeed if their answers satisfy the constraint decorating the edge $uv$.

In the $k$-fold \define{parallel repetition} $G^{\ok}$,
the referee selects $k$ edges $u_1v_1,\ldots,u_kv_k$ independently from $E$
and sends a question tuple $u_1,\ldots,u_k$ to
Alice, and $v_1,\ldots,v_k$ to Bob.
Each player responds with a $k$-tuple of answers and they succeed if
their answers satisfy each of the $k$ constraints on these edges.

\medskip
Parallel repetition is a basic product operation on games, and yet its
effect on the game value  is far from obvious. Contrary to what one might expect, there are strategies for the repeated game that do significantly better than the naive strategy answering each of the $k$ questions using the best single-shot strategy.
Nevertheless, the celebrated parallel repetition theorem of Raz~\cite{Raz98}
bounds the value of $G^{\ok}$ by a function of the value of $G$ that
decays exponentially with the number of repetitions. The broad impact of this theorem can be partly attributed to the general nature of parallel repetition. It is an operation that can be applied to any game without having to know almost anything about its structure.
Raz's proof has since been simplified, giving stronger and sometimes
tight bounds~\cite{Holenstein09, Rao11}. Still, there is much that is left unknown regarding the behavior of games under parallel repetition. For example, previous to this work, it was not known if repetition causes a decrease in the value for a game whose value is already small, say sub-constant. It was also not known how to bound the value of the product of just two games. Other open questions include bounding the value of games with more than two players, and bounding the value of entangled games (where the two players share quantum entanglement). The latter question has recently received some attention \cite{DinurSV14, JainPY14,ChaillouxS13}.

\subsection{Our Contribution}
Our main contribution is a new analytical framework for studying
parallel repetitions of projection games. In this framework, we prove
for any \define{projection
game}:\footnote{%
  In a projection game, for any two questions $u$ and $v$ to the
  players and any answer $\beta$ of $Bob$, there exists at most one
  acceptable answer $\alpha$ for Alice.}
\begin{theorem}[Parallel Repetition Bound]\label{thm:parrep}
  Let $G$ be a projection game. If $\val(G)\le \rho$ for some $\rho>0$ then,
\begin{displaymath}
    \val(G^\ok) \le \left(\frac{2\sqrt{\rho}}{1+\rho} \right)^{k/2}\mper
  \end{displaymath}
\end{theorem}
We remark that for values of $\val(G)$ close to $1$ this theorem matches Rao's bound for projection games, with improved constants.
\begin{corollary}[Parallel Repetition for high-value games \cite{Rao08b}]\label{cor:rao}
  For any projection game $G$ with $\val(G) \le 1-\eps\mcom$
  \begin{displaymath}
    \val(G^\ok) \le \left(1-\ffrac{\eps^2}{16}\right)^{k}\mper
  \end{displaymath}
\end{corollary}
We give a particularly short proof of a (strong) parallel repetition
bound for a subclass of games, namely expanding projection games.
This class of games is rich enough for the main application of
parallel repetition:
NP-hardness of \labelcover with perfect completeness
and soundness close to $0$
(the starting point for most hardness of approximation results).
See \pref{sec:expanding-games}.

Next, we list some new results that are obtained by studying parallel
repetition in this framework.

\paragraph{Repetition of small-value games}

Our first new result is that if the initial game $G$ has value $\rho$
that is possibly sub-constant, then the value of the repeated game
still decreases exponentially with the number of repetitions. Indeed, \pref{thm:parrep} for small $\rho$ becomes,
\begin{corollary}[Repetition of games with small value]\label{cor:parrep-small-soundness}
  For any projection game $G$ with $\val(G) \le \rho\mcom$
  \begin{displaymath}
    \val(G^\ok) \le (4\rho)^{k/4}\mper
  \end{displaymath}
\end{corollary}
This bound allows us to prove NP-hardness for \labelcover
that is better than was previously known (see \pref{thm:labelcover-new}) by applying our small-value parallel repetition theorem on the PCP of \cite{MR10}.
A concrete  consequence is the following corollary.
\begin{corollary}[NP-hardness for \labelcover]\label{cor:labelcover}
 For every constant $c>0$, given a \labelcover instance of size $n$ with alphabet size at most $n$, it is NP-hard to decide if its value is $1$ or at most $\eps = \frac 1 {(\log n)^c}$.
\end{corollary}
\paragraph{Hardness of \setcover}
A famous result of Uriel Feige \cite{Feige-setcover} is that unless $NP\subseteq DTIME(n^{O(\log\log n)})$ there is no polynomial time algorithm for approximating \setcover to within factor $(1-o(1))\ln
n$. Feige's reduction is slightly super-polynomial because it involves, on top of the basic PCP theorem, an application of Raz's parallel repetition theorem with $\Theta(\log\log n)$ number of repetitions. Later, Moshkovitz and Raz \cite{MR10} constructed a stronger PCP whose parameters are closer but still not sufficient for lifting Feige's result to NP-hardness. Moshkovitz \cite{Moshkovitz12} also generalized Feige's reduction to work from a generic projection \labelcover rather than the specific one that Feige was using. Our \pref{cor:labelcover} makes the last step in this sequence of works and gives the first tight NP-hardness for approximating \setcover.
\begin{corollary}[Tight NP-hardness for approximating
  \setcover]\label{cor:setcover}
  For every $\alpha>0$, it is NP-hard to approximate \setcover to within
  $(1-\alpha)\ln n$, where $n$ is the size of the instance. The reduction runs in time $n^{O(1/\alpha)}$.
\end{corollary}
\noindent
Unlike the previous quasi NP-hardness results for \setcover,
\pref{cor:setcover} rules out that  approximation ratios of
$(1-\alpha)\ln n$ can be achieved in time $\cramped{2^{n^{o(1)}}}$ (unless $\mathrm{NP}\sse \mathrm{TIME}(\cramped{2^{n^{o(1)}}})$).
Together with the best known approximation algorithm for \setcover \cite{CyganKW09}, we can characterize the time vs. approximation trade-off for the problem:
\begin{corollary}
  Assuming $\mathrm{NP}\not\sse \mathrm{TIME}(\cramped{2^{n^{o(1)}}})$, the time complexity of achieving an approximation ratio $(1-\alpha)\ln n$ for \setcover is
  \begin{math}
    \cramped{2^{n^{\Theta(\alpha)}}}.
  \end{math}
\end{corollary}

Going back to \labelcover, we remark that the hardness proven in \pref{cor:labelcover} is still far from the
known algorithms for \labelcover and it is an interesting open
question to determine the correct tradeoff between $\eps$ and the alphabet size.

\paragraph{Few repetitions}
Most parallel repetition bounds are tailored to the case that the number of repetitions $k$ is large compared to $1/\eps$ (where as usual, $val(G)=1-\eps$).
For example, when the number of repetitions $k\ll 1/\e$, the bound $\val(G^{\ok})\le (1-O(\e^2))^k\approx (1-O(k\e^2))$ for projection games~\cite{Rao11} is weaker than the trivial bound $\val(G^{\ok})\le 1-\e$.
The following theorem gives an improved and tight bound when $k\ll 1/\e^2$.
\begin{theorem}[Few repetitions]
\label{thm:few-repetitions}
  Let $G$ be a projection game with $\val(G) = 1-\eps$.
  Then for all $k\ll 1/\e^2$,
  \begin{displaymath}
    \val(G^{\otimes k}) \le 1-\Omega(\sqrt k\cdot \eps)\mper
  \end{displaymath}
\end{theorem}
\noindent

\Dnote{}%
A relatively recent line of work \cite{FeigeKO07a,Raz08a,BarakHHRRS08,BarakRRRS09,RazR12} focused on the question of strong parallel repetition. Namely, given a game $G$ with $\val(G) \le 1-\eps$ is it true that $\val(G^\ok)\le (1-O(\eps))^k$? If true for any unique game $G$, such a bound would imply a reduction from \maxcut to unique games. However, Raz \cite{Raz08a} showed that the value of the odd cycle xor game is at least $1-O(\sqrt k \cdot \eps)$, much larger than the $1-O(k\eps)$ in strong parallel repetition. Our bound matches Raz's bound even for small values of $k$, thereby confirming a
conjecture of Ryan O'Donnell\footnote{Personal communication, 2012.} and extending the work of \cite{FeigeKO07a} who proved such an upper bound for the odd-cycle game.
\subsection{High Level Proof Overview}
We  associate a game $G$ with its \emph{label-extended graph}, by blowing up each vertex in $U$ or in $V$ to a cloud of $\card\Sigma$ vertices. In this bipartite graph we connect a left vertex $(u,\alpha)$ in the cloud of $u$ to a right vertex $(v,\beta)$ in  the cloud of $v$ iff $(\alpha,\beta)$ satisfy the constraint attached to $u,v$. This graph naturally gives rise to a linear operator mapping functions $f$ on the right vertices to functions $Gf$ on the left vertices, where $Gf(u,\alpha)$ is defined by aggregating the values of $f(v,\beta)$ over all neighboring vertices.

In this language, the product operation on games is given by the tensor product operation on the corresponding linear operators.

It turns out that a good measure for the value of the game $G$ is its {\em collision value}, which we denote by $\norm G$, which can be viewed as the (square root of the) maximal value of the following symmetrized version of $G$: The referee chooses a question $u$ and two neighbors $v,v'$ of it. It asks Bob for an answer for $v$ and Bob' for an answer for $v'$. The referee accepts if both answers ``collide'', i.e. project to the same answer for $u$. If both Bob and Bob' play according to a strategy $f$ then their value is $\snorm{Gf}$, hence our notation for the collision value, $\norm G$. This value has been studied before, for example in PCP constructions when moving from a line versus point test to a line versus line test. A simple Cauchy--Schwarz inequality shows that (see \pref{claim:collision-val})
\begin{equation}\label{eq:col} \val(G) \le \norm G \le \val(G)^{1/2} \mper \end{equation}
The collision value is more amenable to analysis, and indeed, our main technical theorem shows that the collision value has the following very nice property,
\begin{theorem}\label{thm:collision-val-product}
Any two projections games $G$ and $H$ satisfy
\begin{math}
  \snorm{G\otimes H} \le \varphi(\snorm{G}) \cdot \snorm H\mcom
\end{math}
where $\varphi(x) = \frac{2\sqrt x}{1+x}$.
\end{theorem}
\pref{thm:parrep} follows directly from repeated applications of this theorem, together with the bound in \pref{eq:col}.
We remark that a similar statement for $\val(\cdot)$ instead of $\snorm {\cdot}$, i.e. of the form $\val(G\ot H) \le \varphi(\val(G))\cdot\val(H)$, is false. Feige shows~\cite{Feige91} an example of a game for which $\val(G\ot G)=\val(G)=\frac 1 2$, see discussion in~\pref{app:Feige}. In contrast, \pref{thm:collision-val-product} implies that there is no projection game $G$ for which $\norm{G\ot G} = \norm{G} < 1$.

The proof of \pref{thm:collision-val-product} proceeds by looking for a parameter $\rho_G$ for which $\norm{G\ot H}\le\rho_G\cdot\norm H$, and such that $\rho_G$ depends only on $G$. One syntactic possibility is to take $\rho_G$ to be the supremum of $\frac{\norm{G\ot H}}{\norm H}$ over all games $H$,\footnote{Goldreich suggests to call this the ``environmental value'' of a game $G$ because it bounds the value of $G$ relative to playing in parallel with any environment $H$} but from this definition it is not clear how to prove that $\rho_G\approx \val(G)$.

Instead, starting from this ratio we arrive at a slightly weaker relaxation, which we call $\rval$, that is expressed as the ratio of two collision values: one for the game~$G$ and the other involving a trivial game~$T$ that provides suitable normalization.
In order to connect the ratio $\frac{\norm{G\ot H}}{\norm H}$ to a quotient of the values of $G$ and of $T$, we factor the operator $G\ot H$ into two consecutive steps \[ G\ot H = (G\ot \Id)(\Id \ot H)\mper \] The same factorization can be applied to the operator $T\ot H$.
This factorization allows us to magically ``cancel out'' the game $H$, and we are left with an expression
that depends only on $G$.\\

The main technical component of the proof is to prove that $\rval(G)\approx \val(G)$, and this proof has two components.
The first is a rounding algorithm that extracts an assignment from a non-negative function with large ratio between the norm of $Gf$ and the norm of $Tf$. In the case of expanding games, this is the only component necessary and it is obtained rather easily from the expander mixing lemma.
For non-expanding games, we rely on a subtler (more ``parameter-sensitive'') proof, building on a variant of Cheeger's inequality in \cite{Steurer10c}.
Here a non-negative function will only give a good {\em partial assignment}: an assignment that assigns values only to a small subset of the vertices.
We then combine many partial assignments into a proper assignment using \emph{correlated sampling}, first used in this context by \cite{Holenstein09}.

\subsection{The relaxation $\rval$}

Our definition of $\rval$ involves moving from randomized assignments $f:V\times\Sigma\to\Rnn$ to vector assignments $f:V\times\Sigma\to\Rnn^\Omega$ for some finite measure space $\Omega$.
The collision value of a vector assignment is the expectation over $\omega\in \Omega$ of $\snorm{Gf_\omega}$ where $f_\omega:V\times\Sigma\to\Rnn$ is the $\omega$-th component of $f$. This can be written as $\snorm{(G\ot \Id_\Omega)f}$ where $\Id_\Omega$ is the identity operator on $\R^{\Omega}$.

The relaxed value $\rval(G)$ is the maximum of $\snorm{(G\ot \Id_\Omega)f}$ over all $f$ that are ``normalized'':
 \[ \rval(G) = \sup_\Omega \sup_{f\in Assign(\Omega)} \norm{(G\ot \Id)f}
\]
where the normalization requirement, in analogy to demanding $\sum_\beta f(v,\beta)\le 1$ for randomized assignments, is to require that $f$ belong to the following set
\[
Assign(\Omega) = \Set{ f:V\times\Sigma\to\Rnn^\Omega\; :\; \forall v\in V,\; \Bigsnorm{\sum_\beta f(v,\beta) }\le 1 } \mper
\]
Note that this set coincides with the set of randomized assignments when $\Omega= \set 1$, so clearly $\rval(G)\ge \norm G$.

We remark that the set $Assign(\Omega)$ is nothing but the set of all vector assignments $f$ such that $\snorm{(T_v\ot \Id_\Omega)f}\le 1$ for all $v$, where $T_v$ is a specific ``trivial'' game that is essentially defined by the condition above (defined formally in \pref{sec:T}). So $\rval$ can be written as a kind of  \emph{generalized Rayleigh quotient}\footnote{Rayleigh quotients refer to the expressions of the form $\langle  f,A f \rangle/\langle  f,f \rangle$ for operators $A$ and functions $f$.
\emph{Generalized Rayleigh quotients} refer to expressions of the form $f\mapsto \langle  f,A f \rangle/\langle  f,B f \rangle$ for operators $A$ and $B$ and functions $f$.  },
\[ \rval(G) = \sup_\Omega \sup_{f\ge 0}\frac{\norm{(G\ot \Id)f}}{\max_v \norm{(T_v\ot \Id)f} }\mper
\]
The key difference to standard linear-algebraic quantities is that we restrict this generalized Rayleigh quotient to functions that take only {\em nonnegative values}.
Thus, intuitively one can think of the value $\rval(G)$ as a ``positive eigenvalue'' of the game~$G$.

With this definition in hand it is rather straightforward to prove that $\rval$ is multiplicative
and a good approximation of the game value,
\begin{lemma}[Multiplicativity]\label{lem:mult} For every two projection games $G,H$,
\[\rval(G\ot H) = \rval(G)\cdot\rval(H). \]
\end{lemma}
\begin{theorem}[Approximation]
\label{thm:approximation}
  Let $G$ be a game with $\rval(G)^2 > \rho$, then
  \begin{displaymath}
    \val(G)\geq\snorm G > \frac{1-\sqrt{1-\rho^2}}{1+\sqrt{1-\rho^2}}\mper
  \end{displaymath}
Contrapositively, if $\snorm G < \delta$, then
  \begin{displaymath}
    \rval(G)^2 < \frac{2\sqrt{\delta}}{1+\delta}\mper
  \end{displaymath}
  In particular, if $\snorm G$ is small then the above bound becomes $\rval(G)^2\le 2\norm G$; and if $\snorm G < 1-\eps$ then $\rval(G )^2 < 1-\eps^2/8$.
\end{theorem}
\subsection{Related work}
Already in~\cite{FeigeL92}, Feige and \Lovasz proposed to study parallel
repetition via a relaxation of the game value.
Their relaxation is defined as the optimal value of a semidefinite program.
While this relaxation is multiplicative, it does not provide a good
approximation for the game value.\footnote{%
  In fact, no polynomial-time computable relaxation can provide a good
  approximation for the game value unless $\mathrm P=\mathrm{NP}$,
  because the game value is NP-hard to approximate.}
In particular, the value of this relaxation can be $1$, even if the game
value is close to $0$.
The proof that the \define{Feige--\Lovasz relaxation} is multiplicative uses
semidefinite programming duality (similar to \cite{Lovasz79}).
In contrast, we prove the multiplicativity of $\rval$ in a direct way.\footnote{
  The relaxation $\rval$ can be defined as a convex program (in fact, a copositive program).
  However, it turns out that unlike for semidefinite programs, the
  dual objects are not closed under tensor products.
}

For unique two-player games, Barak et al. \cite{BarakHHRRS08}
introduced a new relaxation, called \define{Hellinger value}, and
showed that this relaxation provides a good approximation to both the game
value and the value of the Feige--\Lovasz relaxation
(see \cite{Steurer10b} for improved approximation bounds).
These quantitative relationships between game value, Hellinger value,
and the Feige--\Lovasz relaxation lead to counter-examples to ``strong
parallel repetition,'' generalizing \cite{Raz08a}.

The relaxation $\rval$ is a natural extension of the Hellinger value
to projection games (and even, general games).
Our proof that $\rval$ satisfies the approximation property for
projection games follows the approach of \cite{BarakHHRRS08}. The proof
is more involved because, unlike for unique games, $\rval$ is no longer
easily expressed in terms of Hellinger distances.
Another difference to \cite{BarakHHRRS08} is that we need to establish
the approximation property also when the game's value is close to $0$.
This case turns out to be related to Cheeger-type inequalities in the
near-perfect expansion regime \cite{Steurer10c}.

\subsection{Organization}
In \pref{sec:technique} we describe the analytic framework in which we
study games. In \pref{sec:basic-approach} we outline the main approach and
show how \pref{thm:collision-val-product} implies \pref{thm:parrep}, \pref{cor:parrep-small-soundness} and \pref{cor:rao} together.
We then give a complete and relatively simple analysis of the parallel repetition bound for expanding projection
games. This proof gives a taste of our techniques in a simplified
setting, and gives gap amplification for \labelcover, thereby proving
the NP-hardness of \labelcover$(1,\delta)$.
In \pref{sec:approx-gen} we prove the
approximation property of $\rval$ for non-expanding games and then multiplicativity of $\rval$ and then derive \pref{thm:collision-val-product}. In
\pref{sec:few-reps} we analyze parallel repetition with few
repetitions, proving \pref{thm:few-repetitions}.
We prove \pref{cor:labelcover} and related
hardness results for \labelcover and \setcover in \pref{sec:inapprox}.

\section{Technique}\label{sec:technique}

\subsection{Games and linear operators}
\label{sec:label-cover-games}

A \emph{two-prover game} $G$ is specified by a bipartite graph with vertex sets $U$ and $V$ and edges decorated by constraints $\pi\sse \Sigma \times \Sigma$ for an alphabet $\Sigma$.
(We allow parallel edges and edges with nonnegative weights.)
The graph gives rise to a distribution on triples $(u,v,\pi)$ (choose an edge of the graph with probability proportional to its weight).
The marginals of this distribution define probability measures on $U$ and $V$.
When the underlying graph is regular, these measures on $U$ and $V$ are uniform.
It's good to keep this case in mind because it captures all of the difficulty.
We write $(v,\pi)\mid u$ to denote the distribution over edges incident to a vertex $u$.
(Formally, this distribution is obtained by selecting a triple $(u,v,\pi)$ conditioned on $u$)

We say that $G$ is a \emph{projection game} if every constraint $\pi$ that appears in the game is a projection constraint, i.e., each $\beta\in \Sigma$ has
at most one $\alpha\in\Sigma$ for which $(\alpha,\beta)\in \pi$.
We write $\alpha\stackrel \pi{\mapsfrom} \beta$ to denote $(\alpha,\beta)\in \pi$ for projection constraints $\pi$.
If the constraint is clear from the context, we just write $\alpha\mapsfrom \beta$.

\Dcomment{needed at this point?:
The computational problem of finding the value of a given two-prover game is called \labelcover.
Concretely, $\labelcover(1,\delta)$ is the problem of distinguishing if a given projection game $G$ has $\val(G)=1$ or $\val(G)\le \delta$.
}

\paragraph{Linear-algebra notation for games}
In this work we will represent an \emph{assignment} for Bob by a nonnegative function $f\from V\times\Sigma \to \R$ such that $\sum_{\beta\in\Sigma}f(v,\beta)=1$ for every $v\in V$.
The value $f(v,\beta)$ is interpreted as the probability that Bob answers $\beta$ when asked $v$.
Similarly, an assignment for Alice is a nonnegative function $g\from U\times \Sigma \to \R$ such that $\sum_\alpha g(u,\alpha)=1$ for all $u\in U$.

Let $L(U\times \Sigma)$ be the space of real-valued functions on $U\times \Sigma$ endowed with the inner product $\iprod{\cdot,\cdot}$,
\begin{displaymath}
  \iprod{g,g'} = \E_u \sum_\alpha g(u,\alpha)g'(u,\alpha).
\end{displaymath}
The measure on $U$ is the probability measure defined by the graph; the measure on $\Sigma$ is the counting measure.
More generally, if $\Omega$ is a measure space, $L(\Omega)$ is the space of real-valued functions on $\Omega$ endowed with the inner product defined by the measure on $\Omega$.
The inner product induces a norm with $\norm{g} = \iprod{g,g}^{1/2}$ for $g\in L(\Omega)$.

We identify a projection game $G$ with a linear operator from $L(V\times\Sigma)$ to $L(U\times \Sigma)$, defined by
\begin{displaymath}
  Gf ( u, \alpha ) = \E_{(v,\pi)|u} \sum_{\beta:\,\alpha\mapsfrom \beta} f(v,\beta)\mper
\end{displaymath}
The bilinear form $\iprod{\cdot, G \cdot}$ measures the value of assignments for Alice and Bob.
\begin{claim}
  If Alice and Bob play the projection game $G$ according to assignments $f$ and $g$, then their success probability is equal to $\langle g, Gf \rangle$,
  \begin{displaymath}
    \iprod{g,Gf}
    = \E_{(u,v,\pi)}\sum_{a\mapsfrom \beta} g(u,\alpha)f(v,\beta) \mper
  \end{displaymath}
\end{claim}
This setup shows that the value of the game $G$ is the maximum of the bilinear form $\langle  g,Gf \rangle$ over assignments $f$ and $g$.
If we were to maximize the bilinear form over all functions with unit norm (instead of assignments), the maximum value would be the largest singular value of an associated matrix.

\subsection{Playing games in parallel}
Let $G$ be a projection game with vertex sets $U$ and $V$ and alphabet $\Sigma$.
Let $H$ be a projection game with vertex sets $U'$ and $V'$ and alphabet $\Sigma'$.
The \emph{direct product} $G\ot H$ is the following game with vertex sets $U\times U'$ and $V\times V'$ and alphabet $\Sigma\times \Sigma'$:
The referee chooses $(u,v,\pi)$ from $G$ and $(u',v',\pi')$ from $H$ independently.
The referee sends $u,u'$ to Alice and $v,v'$ to Bob.
Alice answers $\alpha,\alpha'$ and Bob answers $\beta,\beta'$.
The players succeed if both $\alpha \mapsfrom \beta$ and $\alpha'\mapsfrom \beta'$.
In linear algebra notation,
\begin{claim}
Given two games $G\from L(V\times \Sigma)\to L(U\times\Sigma)$ and $H\from L(V'\times \Sigma)\to L(U'\times\Sigma)$, the direct product game
\begin{math}
  G\otimes H\from L(V\times V' \times \Sigma\times\Sigma')\to L(U\times U' \times \Sigma\times\Sigma')
\end{math}
is given by the tensor of the two operators $G$ and $H$.
More explicitly, for any $f\in L(V\times V' \times\Sigma\times\Sigma')$,
the operator for $G\ot H$ acts as follows,
\[
(G\otimes H) f (u,u',\alpha,\alpha') = \E_{(v,\pi)|u} ~ \E_{(v',\pi')|u'} ~ \sum_{\beta\from (\alpha,\beta)\in \pi} ~ \sum_{\beta'\from (\alpha',\beta')\in \pi'} f(v,v',\beta,\beta')
\] \qed
\end{claim}

The notation $G^{\otimes k}$ is short for $G\otimes \cdots \otimes G$ ($k$ times).

\subsection{The collision value of a game}
\newcommand{\gnorm}[1]{\snorm{#1}}

The collision value of a projection game $G$ is a relaxation of the value of a game that is obtained by moving from $G$ to a symmetrized version of it.\footnote{Such a transformation is well-known in e.g. the PCP literature, albeit implicitly. Moving to a symmetrized version occurs for example in low degree tests when when converting a line versus point test to a line versus line test.}
The advantage of the collision value is that it allows us to eliminate one of the players (Alice) in a simple way.
Let the collision value of an assignment $f$ for Bob be $\norm{Gf} = \iprod{Gf,Gf}^{1/2}$.
We define the \emph{collision value} of $G$ to be \[
\norm{G} = \max_f \norm{Gf}
\]
where the maximum is over all assignments $f\in L(V\times \Sigma)$.

The value $\norm {Gf}^2$ can be interpreted as the success probability of the following process: Choose a random $u$ and then choose independently $(v,\pi)|u$ and $(v',\pi')|u$; choose a random label $\beta$ with probability $f(v,\beta)$ and a random label $\beta'$ with probability $f(v',\beta')$, accept if there is some label $\alpha$ such that $\alpha \stackrel \pi \mapsfrom \beta$ and  $\alpha \stackrel {\pi'} \mapsfrom \beta'$ (in this case $\beta,\beta'$ 'collide').

\inote{}

The collision value of a projection game is quadratically related to its value.
This claim is the only place where we use that the constraints are projections.
\begin{claim}\label{claim:collision-val}
  Let $G$ be a projection game.
  Then $\val(G) \le \norm{G} \le\val(G)^{1/2}$.
\end{claim}
\begin{proof}

Let $f,g$ be assignments attaining the value of $G$.
The first inequality holds because \[
\val(G) = \iprod{g,Gf} \le \norm{g}\cdot\norm{Gf} \le \norm{Gf} \le \norm{G}\mcom
\] where we used Cauchy--Schwarz followed by the bound $\norm{g}\le 1$ which holds for any assignment $g$.
For the second inequality, let $f$ be an assignment for $G$ such that $\norm{Gf}=\norm{G}$.
Then,
\begin{displaymath}
  \gnorm{G}=\snorm{Gf} = \iprod{Gf,Gf} \le \max_g \iprod{g,Gf} =\val(G)
\end{displaymath}
where the inequality used the fact that if $f$ is an assignment then $\sum_\alpha Gf(u,\alpha) \le 1$ for every $u$, so $Gf$ can be turned into a proper assignment $g\in L(U\times\Sigma)$ by possibly increasing some of its entries.
Thus, $\iprod{Gf,Gf}\le \iprod{g,Gf}$ because all entries of these vectors are non-negative.
\end{proof}

The following claim says that the collision value cannot increase if we play, in parallel with $G$, another game $H$.
\begin{claim}\label{claim:monotone}
Let $G,H$ be two projection games, then $\norm{G\ot H}\le \norm {G}$.
\end{claim}
\begin{proof}
This claim is very intuitive and immediate for the standard value of a game, as it is always easier to play one game rather than two, and it is similarly proven here for the collision value. Given an assignment for $G\ot H$ we show how to derive an assignment for $G$ that has a collision value that is at least as high. Let $G$ be a game on question sets $U$ and $V$, and let $H$ be a game on question sets $U'$ and $V'$. Let $f\in L(V\times \Sigma \times V'\times \Sigma')$ be such that $\norm{(G\ot H)f} = \norm{G\ot H}$. For each $v'$ we can define an assignment $f_{v'}$ for $G$ by fixing $v'$ and summing over $\beta'$, i.e. $f_{v'}(v,\beta) := \sum_{\beta'} f(v,v',\beta,\beta')$. For each $u' \in U'$ let $f_{u'}$ be an assignment for $G$ defined by $f_{u'} (v,\beta) := \E_{v'|u'} f_{v'}(v,\beta)$. In words, $f_{u'}$ is the assignment obtained by averaging over $f_{v'}$ for all neighbors $v'$ of $u'$. We claim that $\snorm{G}\ge \E_{u'}\snorm{Gf_{u'}}\ge \snorm{(G\ot H )f}$ where the second inequality comes from the following `coupling' argument:
Select a random vertex $uu'$ and then two possible neighbors $v_1v'_1$ and $v_2v'_2$, and then two answers $\beta_1\beta_1'$ and $\beta_2\beta_2'$ according to $f$. With these random choices a collision for $f$ is if $\beta_1\beta_1'$ is consistent with $\beta_2\beta_2'$. With the same random choices a collision for $f_{u'}$ is when $\beta_1$ is consistent with $\beta_2$, an easier requirement.
\end{proof}

\inote{}

\paragraph{Symmetrizing the Game}
An additional way to view the collision value is as the value of a constraint satisfaction problem (CSP) that is obtained by symmetrizing the game.
\begin{definition}[Symmetrized Constraint Graph]\label{def:G2} For a projection game $G$ given by distribution $\mu$, we define its {\em symmetrized constraint graph} to be the weighted graph $G_{sym}$ on vertex set $V$, described as a distribution $\mu_{sym}$ given by
 \begin{itemize}
 \item Select $u$ at random (according to $\mu$), and then select $(v,\pi)|u$ and independently $(v',\pi')|u$.
 \item Output $(v,v',\tau)$ where $\tau\subseteq\Sigma\times\Sigma$ consists of all pairs $(\beta,\beta')$ such that there exists some $\alpha$ such that $\alpha\stackrel \pi \mapsfrom \beta$ and $\alpha \stackrel {\pi'} \mapsfrom \beta'$.
 \end{itemize}
\end{definition}
It is standard to define the {\em value} of a deterministic assignment $a:V\to\Sigma$ in a constraint graph by \[ val(G_{sym};a) = \Pr_{(v,v',\tau)\sim \mu_{sym}} [(a(v),a(v'))\in \tau]\mper \]
We extend this to randomized assignments $a$, i.e. when $a(v)$ is a random variable taking values in $\Sigma$. In this case the value of $a$ in $G_{sym}$ is defined as the expectation over the values of $a$. If the randomized assignment is described by a vector $f\in \Rnn^{V\times\Sigma}$ so that $a(v)=\beta$ with probability $f(v,\beta)$, one can check that the value of this randomized assignment is equal to
\begin{equation}\label{eq:sym}
\E_{(v,v',\tau)\sim \mu_{sym}}\sum_{(\beta,\beta')\in\tau} f(v,\beta)f(v',\beta') = \snorm{Gf} \mcom
\end{equation}
which is the square of the collision value.
Thus, the value of the CSP described by $G_{sym}$ (which is, as usual, the maximum value over all possible assignments) is equal to $\snorm{G}$.

We note that the symmetrized game $(G\ot H)_{sym}$ can be obtained as the natural direct product of $G_{sym}$ and $H_{sym}$, giving an alternative (intuitive) proof for \pref{claim:monotone}.

\subsection{Expanding Games}\label{sec:expander}
A game $G$ will be called {\em expanding} if the constraint graph of the symmetrized game $G_{sym}$ is an expander.

Formally let $A$ be defined by the matrix describing the random walk in the graph underlying $G_{sym}$, so $A_{v_1,v_2} = \mu_{sym}(v_2|v_1)$. In words, this is the probability of the next step in the random walk landing in $v_2$ conditioned on being in $v_1$. Observe that for any two vectors $x,y\in L(V)$, $\iprod{x,Ay}=\iprod{Ax,y}$ where the inner product is taken as usual with respect to the measure of $V$, which is also the stationary measure of $A$. This implies that $A$ is diagonalizable with real eigenvalues $1=\lambda_1 \ge \lambda_2 \ge \cdots \lambda_n > -1$. Since $A$ is stochastic the top eigenvalue is $1$ and the corresponding eigenvector is the all $1$ vector.
We define the spectral gap of the graph to be $1- \max(\abs{\lambda_2},\abs{\lambda_n})$.

A game $G$ is said to be {\em $c$-expanding} if the spectral gap of the Markov chain $A$ corresponding to $G_{sym}$ is at least $c$.

\subsection{A family of trivial games}\label{sec:T}

We will consider parameters of games (meant to approximate the game value) that compare the behavior of a game to the behavior of certain trivial games. As games they are not very interesting, but their importance will be for normalization.

Let $T$ be the following projection game with the same vertex set $V$ on both sides  and alphabet $\Sigma$:
The referee chooses a vertex $v$ from $V$ at random (according to the measure on $V$).
The referee sends $v$ to both Alice and Bob.
The players succeed if Alice answers with $1\in \Sigma$, regardless of Bob's answer.
The operator of this game acts on $L(V\times \Sigma)$ as follows,
\begin{displaymath}
  Tf(v,\alpha) =
  \begin{cases}
    \sum_\beta f(v,\beta)& \text{ if $\alpha=1$},\\
    0& \text{ otherwise.}
  \end{cases}
\end{displaymath}
We consider a related trivial game $T_v$ with vertex sets $\{v\}$ and $V$ and alphabet $\Sigma$.
The operator maps $L(V\times \Sigma)$ to $L(\{v\}\times \Sigma)$ as follows: $T_v f(v,\alpha)=T f(v,\alpha)$ for $\alpha\in\Sigma$. The operator $T_v$ ``removes'' the part of $f$ that assigns values to questions other than $v$.

While both $T$ and $T_v$ have the same value $1$ for every assignment, they behave differently when considering a product game $G\ot H$ and an assignment $f$ for it.
In particular, the norm of $(T\ot H) f$ may differ significantly from the norms of $(T_v\otimes H)f$, and this difference will be important.

\section{The Basic Approach}\label{sec:basic-approach}

We will prove parallel repetition bounds for the collision value of projection games.
Since the collision value is quadratically related to the usual value, due to \pref{claim:collision-val}, these bounds imply the same parallel repetition bounds for the usual value (up to a factor of $2$ in the number of repetitions).
We state again our main theorem,\\

\noindent{\bf Theorem~\ref{thm:collision-val-product}.~}{\em
Any two projections games $G$ and $H$ satisfy
\begin{math}
  \snorm{G\otimes H} \le \varphi(\snorm{G}) \cdot \snorm H\mcom
\end{math}
where $\varphi(x) = \frac{2\sqrt x}{1+x}$.
}\\

Let us first show how to use \pref{thm:collision-val-product} to quickly prove \pref{thm:parrep} and \pref{cor:parrep-small-soundness} and \pref{cor:rao} together.
By repeated application of \pref{thm:collision-val-product},
\[
\val(G^\ok)^2 \le \snorm{G^\ok} \le \varphi(\snorm G)\cdot \snorm{G^{\otimes k-1}} \le \ldots \le \varphi(\snorm G)^{k-1}\cdot\snorm G  \le \varphi(\snorm G)^k\mcom
\]
where the first inequality is due to \pref{claim:collision-val}, and the last uses that $x\le \varphi(x)$ for $0\le x\le 1$. This gives  \pref{thm:parrep}.
On the interval $[0,1]$ the function $\varphi$ satisfies $\varphi(x)\le 2\sqrt x$ and  $\varphi(1-\e)\le 1-\e^2/8$.
The first bound implies \pref{cor:parrep-small-soundness} and the second bound implies  \pref{cor:rao}.

\subsection{Multiplicative game parameters}\label{subsec:rval}
\inote{}
\Dnote{}
Let us now motivate and develop the definition of $\rval$. Let $G$ be a projection game with vertex sets $U$ and $V$ and alphabet $\Sigma$.
We are looking for a parameter $\rho_G$ of $G$ (namely, a function that assigns a non-negative value to a game) that approximates the value of the game and such that for every projection game $H$,
\begin{equation}\label{eq:mult}
  \norm{G\ot H} \le \rho_G \cdot\norm{H} \mper
\end{equation}
Dividing both sides by $\norm H$, we want a parameter of $G$ that upper bounds the ratio $\ffrac{\norm{G\ot H}}{\norm H}$ for all projection games $H$. The smallest such parameter is simply
\begin{equation}\label{eq:env}
\rho_G =  \sup_H \frac{\norm{G\ot H} }{\norm{H}} \mper
\end{equation}
An intuitive interpretation of this value is that it is a kind of ``parallel value" of $G$, in that it measures the relative decrease in value caused by playing $G$ in parallel with any game $H$, compared to playing only $H$.
Clearly $\norm G \le \rho_G$, but the question is whether $\rho_G\approx \norm G$. We will show that this is the case through another game parameter, $\rval(G)$, such that $\rval(G)\ge \rho_G\ge \norm G$ and such that $\rval(G)\approx \norm G$.
First we introduce the game parameter $\fakerval(G)$ which is not a good enough approximation for $\rho_G$ (although it suffices when $G$ is expanding), and then we refine it to obtain the final game parameter, $\rval(G)$.

To lead up to the definition of $\fakerval(G)$ and then $\rval(G)$, let us assume that we know that $\rho_G> \rho$ for some fixed $\rho>0$. This means that there is some specific projection game $H$ for which $\frac{\norm{G\ot H} }{\norm{H}} > \rho$.
Let $H$ be a projection game with vertex sets $U'$ and $V'$ and alphabet $\Sigma'$,
and let $f$ be an optimal assignment for $G\ot H$, so that $\norm{(G \otimes H)f}=\norm{G\otimes H}$.
We can view $f$ also as an assignment for the game $T\otimes H$, where $T$ is the trivial game from \pref{sec:T}.

The assignment $f$ satisfies $\norm{(T\otimes H)f}\le \norm{T\otimes H}\le \norm{H}$ (the second inequality is by \pref{claim:monotone}).
So
\begin{displaymath}
\rho < \frac {\norm{G\ot H}} {\norm{H}}\le \frac{\norm{(G\ot H) f}}{\norm{(T\ot H) f}}\mper
\end{displaymath}
We would like to ``cancel out'' $H$, so as to be left with a quantity that depends only on $G$ and not on $H$.
To this end, we consider the factorizations $G\otimes H = (G\otimes \Id)(\Id \otimes H)$ and $T\otimes H = (T\otimes \Id)(\Id \otimes H)$, where $\Id$ is the identity operator on the appropriate space.
Intuitively, this factorization corresponds to a two step process of first applying the operator $\Id \ot H$ on $f$ to get $h$ and then applying either $G \ot \Id$ or $T\ot \Id$.
\begin{center}
\begin{tikzpicture}
  \matrix (m) [matrix of math nodes,row sep=0em,column sep=7em,minimum width=2em]
  {
    (G\otimes \Id) f  \\
     & h & f \\
    (T\otimes \Id) f  \\
  };
  \path[-stealth]
    (m-2-2) edge node [above] {\small $ G\otimes \Id$} (m-1-1)
            edge node [below] {\small  $T\otimes \Id$} (m-3-1)
    (m-2-3) edge node [above] {\small $\Id \otimes H$} (m-2-2)
    ;
\end{tikzpicture}
\end{center}
If we let $h=(\Id \otimes H) f$, then
\begin{equation}
  \label{eq:mult1}
\rho<  \frac{\norm{G\otimes H}}{\norm H} \le \frac{\norm{(G\otimes H) f}}{\norm{(T\otimes H) f}} = \frac{\norm{(G\otimes \Id)h}}{\norm{(T\otimes \Id) h}}\mper
\end{equation}

By maximizing the right-most quantity in \pref{eq:mult1} over all nonnegative functions $h$, we get a value that does not depend on the game $H$ so it can serve as a game parameter that is possibly easier to relate to the value of $G$ algorithmically. Observe that there is a slight implicit dependence on (the dimensions of) $H$ since the $\Id$ operator is defined on the same space as that of $H$. It turns out though that the extra dimensions here are unnecessary and the maximum is attained already for one dimensional functions $h$, so the identity operator~$\Id$ can be removed altogether. This leads to the following simplified definition of a game parameter
\begin{definition}\label{def:scalar-rval}
For any projection game $G$ let
\[\fakerval(G) =  \max_{h\ge 0} \frac{\norm{Gh}}{\norm{Th}}\mper \]
\end{definition}
\begin{theorem}
\label{thm:simple-bound}
  Any two projection games $G$ and $H$ satisfy
$\norm{G\otimes H} \le \fakerval(G) \cdot  \norm{H}$. Therefore, $\fakerval(G)\ge \rho_G$.
\end{theorem}
Before proving this, we mention again that in general $\fakerval(G)\not\approx \norm G$ which is why we later make a refined definition $\rval(G)$.
\begin{proof}
  By \pref{eq:mult1}, there exists a non-negative function $h\in L(V\times\Sigma\times V'\times\Sigma')$ such that $\norm{G\otimes H}/\norm{H} \le \norm{(G\otimes \Id)h}/\norm{(T\otimes \Id)h}$.
  We can view $h$ as a matrix with each column belonging to $L(V\times\Sigma)$ - the input space for $G$ and $T$.
  Then, $(G\otimes \Id) h$ is the matrix obtained by applying $G$ to the columns of the matrix $h$, and $(T\otimes \Id)h$ is the matrix obtained by applying $T$ to the columns of $h$.
Next, we expand the squared norms of these matrices column-by-column,
\[
 \frac{\snorm{(G\ot \Id)h} }{\snorm{(T\ot \Id)h}} = \frac{\E_j \snorm{Gh_{j}} }{\E_j \snorm{Th_{j}} }
\]
where $j$ runs over the columns (this happens to be $j\in U'\times\Sigma'$ but it is not important at this stage).
An averaging argument implies that there is some $j^*$ for which $\frac{\snorm{Gh_{j^*}} }{\snorm{Th_{j^*}} } $ is at least as large as the ratio of the averages. Removing the squares from both sides, we get
\[ \frac{\norm{G\otimes H}}{\norm H} \le \frac{\norm{Gh_{j^*}}}{\norm{Th_{j^*}}} \le \fakerval(G)\qedhere \]
\end{proof}

 \Dnote{}

In the next subsection we will show that for expanding games $G$, $\fakerval(G)\approx \norm G$, thus proving \pref{thm:collision-val-product} for the special case of expanding games.
For non-expanding games, $\fakerval(G)$ is not a good approximation to $\norm G$ and a more refined argument is called for as follows.
Instead of comparing the value of $G\ot H$ to $T\ot H$, we compare it to the collection $T_v\ot H$ for all $v\in V$ (defined in \pref{sec:T}).
We observe that the inequality \pref{eq:mult1} also holds with $T$ replaced by $T_v$, so that for all $v\in V$
\begin{equation}
  \label{eq:mult-v}
\rho<  \frac{\norm{G\otimes H}}{\norm H} \le \frac{\norm{(G\otimes H) f}}{\norm{(T_v\otimes H) f}} = \frac{\norm{(G\otimes \Id)h}}{\norm{(T_v\otimes \Id) h}}\mper
\end{equation}
Now, by maximizing the right hand side over all measure spaces $L(\Omega)$ for $\Id$ and over all non-negative functions $h\in L(V\times\Sigma\times\Omega)$, we finally arrive at the game parameter
\begin{equation}
\label{eq:def-valplus}
  \rval(G) \defeq \sup_{\Omega}\max_{h\ge 0} \frac{\norm{(G\otimes \Id)h}}{\max_v\norm{(T_v\otimes \Id) h}}\mcom
\end{equation}
where the $\Id$ operator is defined on the measure space $L(\Omega)$ and we are taking the supremum over all finite dimensional measure spaces $\Omega$.
It turns out that the supremum is attained for finite spaces---polynomial in the size of $G$.

\begin{theorem}
\label{thm:valplus-bound}
  Any two projection games $G$ and $H$ satisfy
  \begin{math}
    \norm{G\otimes H} \le \rval(G)\cdot  \norm{H}\mper
  \end{math}
\end{theorem}
\begin{proof}
We essentially repeat the proof for $\fakerval(\cdot)$. Let $H$ be any projection game, and let $f$ be an optimal assignment for $G\ot H$. For every question $v$ we compare $\norm{(G\ot H)f}$ to $\norm{(T_v\ot H)f}$ for $T_v$ the trivial operator from \pref{sec:T}. Since $\norm{(T_v\ot H)f}\le \norm{T_v\ot H}\le \norm H$ we get,
\begin{multline*}
  \frac{\norm{G\ot H}}{\norm H} \le \frac{\norm{(G\ot H)f}}{ \max_v
    \norm{(T_v\ot H)f}} \le \max_{f'\ge 0} \frac{\norm{(G\ot H)f'}}{ \max_v
    \norm{(T_v\ot H)f'}}\\ \le \max_{h\ge 0} \frac{\norm{(G\ot \Id)h}}{
    \max_v\norm{ (T_v\ot \Id)h}} = \rval(G)\mper
\end{multline*}
\end{proof}
The advantage of $\rval$ over $\fakerval$ will be seen in the next sections, when we show that this value is a good approximation of value of the game $G$.
\subsection{Approximation bound for expanding projection games}\label{sec:expanding-games}

In this section we show that if $G$ is an expanding game then $\fakerval(G)\approx \norm G$. (Recall from \pref{sec:expander} that a game is called $\gamma$-expanding if the graph underlying $G_{sym}$ is has spectral gap at least $\gamma$.) By definition, $\norm G\le \fakerval(G) $. The interesting direction is

\Icomment{This is a bipartite graph, let $G^TG$ be its symmetrization with respect to $V$, i.e. the graph ...}
\begin{theorem}\label{thm:approx-expander}
Let $G$ be a $\gamma$-expanding projection game for some $\gamma>0$.
Suppose $\fakerval(G) > 1-\eps$. Then, $\norm G > 1-O(\eps/\gamma)$.
\end{theorem}

\begin{proof}
We may assume $\e/\gamma$ is sufficiently small, say $\e/\gamma\le 1/6$, for otherwise the theorem statement is trivially true.
Let $f \in L(V\times\Sigma)$ be nonnegative, such that
\[ \frac{\norm{Gf}}{\norm {Tf}} >1-\eps \mcom
\]
witnessing the fact that $\fakerval(G)>1-\eps$.
First, we claim that without loss of generality we may assume that $f$ is \emph{deterministic}, i.e., for every vertex $v$, there is at most one label $\beta$ such that $f(v,\beta)>0$.
(This fact is related to the fact that randomized strategies can be converted to deterministic ones without decreasing the value.)
We can write a general nonnegative function $f$ as a convex combination of deterministic functions $f'$ using the following sampling procedure:
For every vertex $v$ independently, choose a label $\beta_v$ from the distribution that gives probability $\frac {f(v,\beta')}{\sum_\beta f(v,\beta)}$ to label $\beta'\in\Sigma$. (If the denominator is zero then choose nothing).
Set $f'(v,\beta_v) = \sum_{\beta} f(v,\beta)$ and $f'(v,\beta')=0$ for $\beta'\neq \beta_v$.
The functions $f'$ are deterministic and satisfy $\E f' = f$.
By convexity of the function $f \mapsto \norm{Gf}$ (being a linear map composed with a convex norm), we have $\E \norm{Gf'} \ge \norm{G \E f'} = \norm {Gf}$.
So, there must be some $f'$ for which $\norm{Gf'}\ge \norm{Gf}$.
By construction, $Tf(v) = \sum_\beta f(v,\beta) = Tf'(v)$, so $\frac {\norm{Gf'}}{\norm{Tf'} }\ge \frac {\norm{Gf}}{\norm{Tf} }$.
(We remark that this derandomization step would fail without the premise that $f$ is nonnegative.)

Thus we can assume that for every vertex $v$, there is some label $\beta_v$ such that $f(v,\beta')=0$ for all $\beta'\neq \beta_v$.
We may also assume that $\lVert  T f \rVert=1$ (because we can scale $f$).
Since $f$ is deterministic, we can simplify the quadratic form $\lVert  G f \rVert^2$,
\begin{equation}
  \label{eq:1-have}
  (1-\e)^2 \le \lVert  G f \rVert^2 =
  \E_{(v,v')} f(v,\beta_v) f(v',\beta_{v'}) \cdot Q_{v,v'},
\end{equation}
where the pair $(v,v')$ is distributed according to the edge distribution of the symmetrized game $G_{\mathrm{sym}}$
and $Q_{v,v'}=\Pr_{\tau \mid (v,v')}\set{(\beta_v,\beta_{v'})\in \tau}$.
Since we can view $b(v)\seteq \beta_v$ as an assignment for $G_{\mathrm{sym}}$, we can lower bound the value of the symmetrized game in terms of $Q_{v,v'}$,
\begin{equation}
  \label{eq:2-want}
  \lVert  G \rVert^2=\val(G_{sym})\ge \val(G_{sym},b) = \E_{(v,v')} Q_{v,v'}\mper
\end{equation}
To prove the theorem, we will argue that the right hand sides of \pref{eq:1-have} and \pref{eq:2-want} are close.
The key step toward this goal is to show that $f(v,\beta_v)\approx 1$ for a typical vertex $v$. A priori $f(v,\beta_v)$ can be very high for some vertices, and very small for others, while maintaining $\E_v f(v,\beta_v)^2=1$. However, the expansion of $G$ will rule this case out.

Denote $g(v) = f(v,\beta_v)$ and let $A$ be the matrix corresponding to the Markov chain on the graph underlying $G_{sym}$ and let $L = Id - A$. The smallest eigenvalue of $L$ is 0 corresponding to the constant functions, and the second smallest eigenvalue is at least $\gamma>0$ because of the expansion of $G$. Now, using inner products with respect to the natural measure on $V$,
\[
\iprod{g,Lg} = \E_v f(v,\beta_v)^2 - \E_{(v,v')}f(v,\beta_v)f(v',\beta_{v'}) \le 1 - (1-\eps)^2 \le 2\eps\mper
\]
On the other hand if we write $g = \bar f+ g^\perp$ for $\bar f = \E_v f(v,\beta_v) =\E_v g(v)$, we have $\iprod{g^\perp,\bar f}=0$ and
\[ \iprod{g,Lg} = \iprod{\bar f,L\bar f} +\iprod{g^\perp,Lg^\perp} = 0+ \iprod{g^\perp,Lg^\perp}  \ge \gamma \snorm{g^\perp}\mper
\]
Combining the above we get $\E_v (f(v,\beta_v)-\bar f)^2 = \snorm{g^\perp}\le 2\eps/\gamma$ which means that $g \approx \bar f \approx 1$.
At this point we could write, using \pref{eq:1-have},
\[
    (1-\eps)^2 \le \E_{(v,v')} f(v,\beta_v) f(v',\beta_{v'}) \cdot Q_{v,v'} =
    \E_{(v,v')} \bar f \cdot \bar f\cdot Q_{v,v'} + err \] where
    \[err = \frac 1 2\left ( \E_{(v,v')} f(v,\beta_v)(f(v',\beta_{v'})-\bar f )Q_{v,v'} + \E_{(v,v')} (f(v,\beta_v)-
    \bar f)f(v',\beta_{v'})Q_{v,v'}\right)
\]
can be bounded by $\norm{g^\perp}$ using the Cauchy-Schwarz inequality.
This already gives a meaningful lower bound of $\snorm{G} = \E Q_{v,v'} \ge 1-O(\sqrt{\eps/\gamma})$, but not the ``strong'' lower bound of $1-O(\eps/\gamma)$ that we are after.
Instead, we conclude the proof of the theorem with the following chain of inequalities,
\begin{align*}
  1-\lVert G \rVert ^2
  & \le \E_{(v,v')} (1-Q_{v,v'}) \quad\using{\pref{eq:2-want}}\\
  & \le \E_{(v,v')} (1-Q_{v,v'}) \cdot 9\Bigparen{\bigparen{f(v,\beta_v)-\bar f}^2 + \bigparen{f(v',\beta_{v'})-\bar f}^2 + f(v,\beta_v) f(v',\beta_{v'}) }\\
  & \le 36 \e/\gamma + 9\cdot  \E_{(v,v')} (1-Q_{v,v'}) \cdot f(v,\beta_v) f(v',\beta_{v'})\\
  & \le 36 \e/\gamma + 18 \e\mper
\end{align*}
The second step uses that all nonnegative numbers $a$ and $b$ satisfy the inequality $1 \le 9 a b + 9\paren{a-\bar f}^2 + 9\paren{b-\bar f}^2$ (using $\bar f\ge 2/3$).
To verify this inequality we will do a case distinction based on whether $a$ or $b$ are smaller than $1/3$ or not.
If one of $a$ or $b$ is smaller than $1/3$, then one of the last two terms contributes at least $1$ because $\bar f\ge 2/3$.
On the other hand, if both $a$ and $b$ are at least $1/3$, then the first term contributes at least $1$.
The third step uses the $f(v,\beta_v)$ is close to the constant function $\bar f\cdot \Ind$.
The fourth step uses \pref{eq:1-have} and the fact that $\E_{(v,v')} f(v,\beta_v)f(v',\beta_{v'})\le  \E_v f(v,\beta_v)^2=1$.
\end{proof}

\subsection{Short proof for the hardness of \labelcover}\label{sec:labelcover-short}

$\labelcover(1,\delta)$ is the gap problem of deciding if the value of a given projection game is $1$ or at most $\delta$. The results of this section suffice to give the following hardness of \labelcover, assuming the PCP theorem. This result is a starting point for many hardness-of-approximation results.
\begin{theorem*}
$\labelcover(1,\delta)$ is NP-hard for all $\delta>0$.
\end{theorem*}
Let us sketch a proof of this. The PCP theorem~\cite{AS,ALMSS} directly implies that $\labelcover(1,1-\eps)$ is NP-hard for some constant $\eps>0$. Let $G$ be an instance of $\labelcover(1,1-\eps)$. We can assume wlog that $G$ is expanding, see~\pref{claim:expanding}. We claim that $G^\ok$ for $k= O(\frac{\log 1/\delta  }{\eps})$ has the required properties. If $\val(G)=1$ clearly $\val(G^\ok)=1$. If $\val(G) < 1-\eps$, then
\[ \val(G^\ok) \le \norm{G^\ok} \le \fakerval(G)^k \le (1-\Omega(\eps))^k \le \delta\]
where the second inequality is due to repeated applications of \pref{thm:simple-bound}, and the third inequality is due to \pref{thm:approx-expander}.
\qed

\section{General projection games}
\label{sec:approx-gen}

The bulk of this section is devoted to proving the approximation property of $\rval$ (\pref{thm:approximation}). Towards the end of the section we also quickly prove the multiplicativity of $\rval$ (\pref{lem:mult}). We end by proving \pref{thm:collision-val-product}.

\subsection{Approximation bound for (non-expanding) projection games}
We prove\\
\noindent {\bf Theorem~\ref{thm:approximation}.~~}{\em
  Let $G$ be a game with $\rval(G)^2 > \rho$, then
  \begin{displaymath}
    \val(G)\geq\snorm G > \frac{1-\sqrt{1-\rho^2}}{1+\sqrt{1-\rho^2}}\mper
  \end{displaymath}
Contrapositively, if $\snorm G < \delta$, then
  \begin{displaymath}
    \rval(G)^2 < \frac{2\sqrt{\delta}}{1+\delta}\mper
  \end{displaymath}
  In particular, if $\snorm G$ is small then the above bound becomes $\rval(G)^2\le 2\norm G$; and if $\snorm G < 1-\eps$ then $\rval(G )^2 < 1-\eps^2/8$.
}\\~\\
Let $G$ be a projection game with vertex set $U,V$ and alphabet $\Sigma$.
Our assumption that $\rval(G)^2>\rho$ implies the existence of a measure space $\Omega$ and a non-negative function $f\in L(V\times\Sigma\times\Omega)$ such that
\[ \snorm{(G\ot \Id_\Omega)f} > \rho \max_v \snorm{(T_v\ot  \Id_\Omega)f} \]

(In this section we denote by $\Id_\Omega$, rather than $\Id$, the identity operator on the space $L(\Omega)$, to emphasize the space on which it is operating.)

Without loss of generality, by rescaling, assume that $f\le 1$.
Since the integration over $\Omega$ occurs on both sides of the inequality, we can rescale the measure on $\Omega$ without changing the inequality, so that $\max_v \norm{(T_v\ot  \Id_\Omega)f} = 1$.
\paragraph{Proof Overview}
The proof is by an algorithm that extracts from $f$ an assignment for $G$.
We have seen that for expanding games $G$ there is always a single element $\omega\in \Omega$ such that the slice $f_\omega$ obtained by restricting $f$ to $\omega$ is already ``good enough'', in that it can be used to derive a good assignment for $G$ (see \pref{thm:approx-expander}). When $G$ is not expanding this is not true because each $f_\omega$ can potentially concentrate its mass on a different part of $G$. For example, imagine that $G$ is made of two equal-sized but disconnected sub-games $G_1$ and $G_2$, with optimal assignments $f_1,f_2$ respectively, and let $f$ be a vector assignment defined as follows. For every $v,\beta$ we will set $f(v,\beta,\omega_1)=f_1(v,\beta)$ if $v$ is in $G_1$ and $f(v,\beta,\omega_2)=f_2(v,\beta)$ if $v$ is in $G_2$. Everywhere else we set $f$ to $0$. The quality of $f$ will be proportional to the average quality of $f_1,f_2$, yet there is no single $\omega$ from which an assignment for $G$ can be derived.
Our algorithm, therefore, will have to construct an assignment for $G$ by combining the assignments derived from different components ($\omega$'s) of $f$.
Conceptually the algorithm has three steps.
\begin{enumerate}
\item In the first step, \pref{lem:deterministic}, we (easily) convert $f$ into a ``deterministic'' vector assignment, i.e. such that for each $v,\omega$ assigns at most one $\beta$ a non-zero value.
\item In the second step we convert each $f_\omega$ into distribution over partial assignments, namely,  $0/1$ functions that might be non-zero only on a small portion of $G$. We call this a ``Cheeger inequality for projection games'', because it converts fractional assignments to distributions over $0/1$ assignments, via randomized rounding. This is done in \pref{lem:threshold-rounding} and is analogous to but more subtle than \pref{thm:approx-expander}.
\item The third and last step, done in \pref{lem:correlated-sampling}, is to combine the different partial assignments into one global assignment. For this step to work we must ensure that the different partial assignments cover the entire game in a uniform way. Otherwise one may worry that all of the partial assignments are concentrated on the same small part of $G$. Indeed, that would have been a problem had we defined $\rval(G)$ to be
$\sup_{\Omega,f} \tfrac{\snorm{(G\otimes \Id_\Omega)f}}{\snorm {(T\otimes \Id_\Omega)f}} $. Instead, the denominator in the definition of $\rval$ is the maximum over $v$ of $\snorm{(T_v\otimes \Id_\Omega)f}$. This essentially forces sufficient mass to be placed on each vertex $v$ in $G$, so $G$ is uniformly covered by the collection of partial assignments.

This step corresponds to ``correlated sampling'', introduced to this context in \cite{Holenstein09}, because when viewed as a protocol between two players, the two players will choose $\omega$ using shared randomness, and then each player will answer his question according to the partial assignment derived from $f_\omega$.
(To be more accurate, shared randomness is also used in the second step, for deciding the rounding threshold through which a partial assignment is derived).
\end{enumerate}
\paragraph{Terminology} Throughout this section we use the following terminology. A {\em fractional assignment} is a non-negative function $h\in L(V\times\Sigma)$. 
A {\em deterministic fractional assignment} is a fractional assignment that for each $v\in V$ assigns at most one $\beta\in\Sigma$ a non-zero value. A {\em partial assignment} is a deterministic fractional assignment whose values are either $0$ or $1$.

A {\em vector assignment} is a non-negative function $f\in L(V\times\Sigma\times\Omega)$ also viewed as a function $f:V\times\Sigma\to\Rnn^\Omega$. A {\em slice} $f_\omega$ of $f$ is the fractional assignment defined by $f_\omega(v,\beta) = f(v,\beta,\omega)$. $f$ is called a {\em deterministic vector assignment} if every slice $f_\omega$ is deterministic.
\subsubsection{Making $f$ deterministic}
%
\begin{lemma}\label{lem:deterministic}
Let $f\in L(V\times\Sigma\times\Omega)$ be non-negative. Then there is a deterministic vector assignment $f'\in L(V\times\Sigma\times\Omega)$ such that
\begin{itemize}
\item For each $v$, $\snorm{(T_v\ot  \Id_\Omega)f}=\snorm{(T_v\ot  \Id_\Omega)f'}$.
\item $\snorm{(G\ot  \Id_\Omega)f'}\ge\snorm{(G\ot  \Id_\Omega)f}$.
\end{itemize}
\end{lemma}
\begin{proof}
We show that any fractional assignment $h$ can be replaced by a deterministic fractional assignment $h'$ such that $\snorm{T_v h'} = \snorm{T_v h}$ for all $v$ and $\snorm{G h'} \ge \snorm{G h}$. The function $f'$ will be defined by replacing each slice $f_\omega$ by $(f_\omega)'$.

Fix $h\in L(V\times\Sigma)$. For each $v$ let $z(v) = \sum_\beta h(v,\beta)$ and define a distribution over deterministic fractional assignments $\set{h_{r}}_r$ according to the following random process. If $z(v)=0$ let $h_{r}(v,\beta)=0$ for all $\beta$. If $z(v)>0$ choose $\beta$ with probability $h(v,\beta)/z(v)$, let $h_r(v,\beta)$ equal $z(v)$ for that $\beta$ and zero for all other $\beta$. Clearly $h= \E_r h_r$. For each $r$, $h_r$ is a  deterministic fractional assignment such that $\norm{T_v h_r} = z(v) = \norm {T_v h}$ for all $v$ and all $r$.
%
%
By convexity $\E_r\norm{Gh_r} \ge \norm{G\E_r h_r} = \norm{G h}$ so there must be some $r$ for which $\norm { Gh_r}$ is at least as high as the average. Set $h' = h_r$.
\end{proof}
From the claim it is clear that
\[\snorm{(G\ot \Id_\Omega)f'}\ge \snorm{(G\ot \Id_\Omega)f} > \rho \max_v \snorm{(T_v\ot  \Id_\Omega)f} = \rho \max_v \snorm{(T_v\ot  \Id_\Omega)f'}\]
so wlog we may assume from now on that $f$ is a deterministic vector assignment.
\subsubsection{Cheeger inequality for projection games}

\begin{lemma}\label{lem:threshold-rounding}
Let $f\in L(V\times\Sigma\times\Omega)$ be a deterministic vector assignment such that $\snorm{(G\ot \Id_\Omega)f}\ge \rho $ and such that $\max_v \snorm{(T_v\ot\Id_\Omega) f}=1$.  Then there exists a $0/1$-valued function $f'\in L(V\times\Sigma\times\Omega\times [0,1])$, such that for $\psi(x)=1-\cramped{\paren{1-\cramped{x^2}}^{1/2}}$,
  \[ \snorm{(G\ot \Id_\Omega\ot \Id_{[0,1]})f'} > \psi(\rho) \]
  and yet \[\max_v \norm{(T_v\ot  \Id_\Omega\ot \Id_{[0,1]})f'} \le  1 \mper \]
\end{lemma}
We will obtain $f'$ from $f$ by applying the following lemma to each slice $f_\omega$, for all $\omega$  simultaneously.
%
\begin{lemma}[Cheeger inequality for projection games] Let $h \in L(V\times\Sigma)$ be a deterministic fractional assignment such that $\snorm{Gh} = \rho \snorm{h}$. There is a distribution $\set{h_\tau}_\tau$ over partial assignments $h_\tau:V\times\Sigma\to \set{0,1}$ such that for $\psi(x)=1-\cramped{\paren{1-\cramped{x^2}}^{1/2}}$,
\[
\forall v\in V;\quad \E_\tau \snorm{T_v h_\tau} = \snorm{T_v h}\qquad\hbox{and}\qquad  \E_\tau \snorm{Gh_\tau} \ge \psi(\rho)\snorm{h}\mper\]
\end{lemma}

\begin{proof}
  Let $h_\tau:V\times \Sigma\to\set{0,1}$ be defined as
  \begin{displaymath}
    h_\tau(v,\beta) =
    \begin{cases}
      1 & \text{ if }h(v,\beta)^2>\tau\mper\\
      0 & \text{ otherwise.} \\
    \end{cases}
  \end{displaymath}
  Since $h$ is deterministic, each $h_{\tau}$ is also deterministic.
  %
  %
%
  %
  For every $v,\beta$, we have $\E_\tau h_\tau(v,\beta)^2=h(v,\beta)^2$.
  Thus, for each $v\in V$,
\[ \E_\tau \snorm{T_vh_\tau} = \E_\tau \big(\sum_\beta h_\tau(v,\beta)\big)^2 =
\E_\tau \sum_\beta h_\tau(v,\beta)^2 = \sum_\beta h(v,\beta)^2 = \snorm{T_vh}
\]
where the second and the last equalities hold because $h$ and $h_\tau$ are deterministic.
It remains to show that
  \[ \E_\tau \snorm {G h_\tau} \ge \psi(\rho) \snorm h \mper\]
Let $x \in \Sigma^V$ be an assignment that is consistent with $h$ (so that $h(v,\beta)=0$   for all $\beta\neq x_{v}$).
Note that also $h_{\tau}(v,\beta)=0$ whenever $\beta\neq x_{v}$, so
\begin{equation}
  \label{eq:pi}
  \lVert  G h_\tau \rVert^2 =
  \E_{(v_1,v_1,\pi)} \sum_{\beta_1,\beta_2}
  h_\tau(v_1,\beta_{1}) h_\tau(v_2,\beta_{2}) \pi(\beta_1,\beta_2)
  = \E_{(v_1,v_1,\pi)} h_\tau(v_1,x_{v_1}) h_\tau(v_2,x_{v_2}) \cdot \pi(x_{v_1},x_{v_2}) ,
\end{equation}
where $(v_1,v_2,\pi)$ is distributed according to the symmetrized game $G_{\mathrm{sym}}$, see \pref{eq:sym}.

For every pair $v_1,\beta_1$ and $v_2,\beta_2$,
  \begin{align}
    \E_{\tau\sim [0,1]}h_{\tau}(v_1,\beta_1)\cdot h_{\tau}(v_2,\beta_2)
    &=\min\Set{h(v_1,\beta_1)^2,h(v_2,\beta_2)^2}
    \label{eq:min}
  \end{align}
  Now, combining \pref{eq:pi} and \pref{eq:min},
  \begin{align}
    \E_{\tau}\snorm{G h_{\tau}}
    & = \E_{(v_1,v_2,\pi)} \pi(x_{v_1},x_{v_2}) \cdot
    \min\Set{h(v_1,x_{v_1})^2,h(v_2,x_{v_2})^2}\label{eq:intermediate}
    \mper
  \end{align}
We can express $\snorm{G h}$ in a similar way,
  \begin{equation}
    \snorm{Gh}
     = \E_{(v_1,v_2,\pi)} \pi(x_{v_1},x_{v_2}) \cdot
    h(v_1,x_{v_1})h(v_2,x_{v_2})
    \label{eq:intermediate2}
  \end{equation}
  At this point we will use the following simple inequality (see
  \pref{cor:gm-vs-min} toward the end of this section):

  \begin{quote}
    \emph{Let $A,B,Z$ be jointly-distributed random variables, such that
    $A,B$ take only nonnegative values and $Z$ is $0/1$-valued.
    Then, $\E Z \min\set{A,B}\ge \psi(\rho)\E \tfrac12(A+B)$ holds
    as long as $\E Z \sqrt{AB}\ge \rho \E \tfrac12(A+B)$.}
  \end{quote}
  We will instantiate the inequality with $Z=\pi(x_{v_1},x_{v_2})$,
  $A=h(v_1,x_{v_1})^2$, and $B=h(v_2,x_{v_2})^2$ (with $(v_1,v_2,\pi)$
  drawn as above).
  With this setup, $\E \tfrac12(A+B)=\snorm{h}$.
  Furthermore, $\E Z \min\set{A,B}$ corresponds to the right-hand side of
  \pref{eq:intermediate} and $\E Z \sqrt{AB}$ corresponds to the
  right-hand side of \pref{eq:intermediate2}.
  Thus, $\snorm{Gh}\ge \rho\snorm{h}$
  means that the condition above is satisfied, and we get the desired
  conclusion, $\E_\tau \snorm{Gh_\tau}\ge
  \psi(\rho)\snorm{h}$ as required.
 \end{proof}

 We now prove \pref{lem:threshold-rounding},
 \begin{proof}
 For each slice $f_\omega$ of $f$ let $\rho_\omega$ be such that $\snorm{Gf_\omega} =\rho_\omega\snorm {f_\omega}$. For $f_\omega$ we have, by the previous lemma, a distribution of $0/1$ functions  $\set{f_{\omega,\tau}}_\tau$ such that $\E_\tau \snorm{Gf_{\omega,\tau}}\ge \psi(\rho_\omega)\snorm{f_{\omega}}$, and such that for every $v$, $\snorm{T_v f_\omega} = \E_\tau \snorm{T_v f_{\omega,\tau}}$.
Let $f'\in L(V\times\Sigma\times\Omega\times[0,1])$ be defined by $f'(v,\beta,\omega,\tau) = f_{\omega,\tau}(v,\beta)$.
For every $v$,
\[  \snorm {(T_v \ot \Id_{\Omega\times [0,1]})f'} =
\int_\Omega \E_\tau \snorm {T_v f'_{\omega,\tau}} = \int_\Omega \snorm{T_v f_\omega} = \snorm {(T_v \ot \Id_{\Omega})f} \le 1 \mper
 \]
and
  \begin{displaymath}
    \snorm{(G\ot \Id_{\Omega\times[0,1]})f'}=
    \int_\Omega \E_\tau \snorm{Gf'_{\omega,\tau}} \, d \omega
    \ge\int_\Omega \psi(\rho_\omega)\cdot \snorm{f_\omega} \, d\omega
    \ge \psi(\rho) .
  \end{displaymath}
where for the last step, we use the convexity of $\psi$ and that \inote{}
  $\rho=\int_\Omega \rho_\omega\snorm{f_\omega}\,d\omega$ and
  $\int_\Omega \snorm{f_\omega}\, d\omega = \snorm{f} \le 1$. The bound $\snorm f\le 1$ follows by
\[\snorm{f} =  \E_v \sum_\beta \int_\Omega   f(v,\beta,\omega)^2 d\omega
= \E_v \int_\Omega \big(\sum_\beta f(v,\beta,\omega)\big)^2  d\omega
= \E_v \int_\Omega \snorm{T_v f_\omega} d\omega
= \E_v \snorm{(T_v \ot \Id_\Omega) f_\omega}
\le 1
\]
where we have used in the second equality that $f$ is deterministic.

Let $\Omega'=\Omega\times [0,1]$. We now have a $0/1$ function $f'\in L(V\times\Sigma\times \Omega')$ such that \[\snorm{(G\ot \Id_{\Omega'})f'}\ge \psi(\rho) \qquad \hbox{and} \qquad \max_v \snorm{(T_v\ot \Id_{\Omega'})f'} = 1  \]
as required.
\end{proof}

\subsubsection{Combining the partial assignments}
The next step is to combine the different slices $f'_{\omega,\tau}$ of $f'$ into one assignment for $G$.
\Dnote{}%
\begin{lemma}[Correlated Sampling]
\label{lem:correlated-sampling}
  There exists an assignment $x$ for $G$ with value at least
  $\frac{1-\gamma}{1+\gamma}$ for $1-\gamma=\snorm{(G\ot \Id_{\Omega'})f'}\ge \psi(\rho)$.
\end{lemma}

\begin{proof}
  We may assume $\snorm{(T_v\ot \Id)f'}=1$ for all $v\in V$.
  (We can arrange this condition to hold by adding additional points to $\Omega'$, one for each vertex in $V$, and extending $f'$ in a suitable way.)

  Rescale $\Omega'$ to a probability measure.
  Let $\lambda$ be the scaling factor, so that $\snorm{(G\ot I_{\Omega'})f'}=(1-\gamma)\lambda$ (after rescaling).
  Then, $f'$ also satisfies $\snorm{(T_v\ot \Id)f'} = \lambda$ for all $v\in V$.

  For every $\omega\in \Omega'$, the slice $f'_\omega$ is a partial
  assignment (in the sense that it uniquely assigns a label per vertex to a
  subset of vertices).
  We will construct (jointly-distributed) random variables
  $\set{X_v}_{v\in V}$, taking values in $\Sigma$, by combining the
  partial assignments $f'_{\omega}$ in a probabilistic way.
(The sampling procedure we describe corresponds to correlated sampling applied to the distributions $\Omega'_v=\{\omega\in \Omega'\mid f'_{\omega}(v,\beta)>0\hbox{ for some (unique) }\beta \}$.)

  Conceptually we define $X$ by randomly permuting $\Omega'$ and then assigning each $v$ according to the first $f'_{\omega}$ that assigns it a non zero value. Formally, let $\set{\omega(n)}_{n\in\N}$ be an infinite sequence of independent
  samples from $\Omega'$ and consider the infinite sequence of slices $f'_{\omega(1)}, f'_{\omega(2)}, f'_{\omega(3)}, $ and so on.
  For each $v$ let $R(v)$ be the smallest index $r$ such that the partial assignment $f'_{\omega(r)}$ assigns a label to $v$, i.e. such that $f'_{\omega(r)}(v,\beta) >0 $ for some (unique) $\beta\in \Sigma$.
  We define $X_v$ to be the label assigned to $v$ by this procedure, so that $f'_{\omega(R(v))}(v,X_v)= 1$.

  In this randomized process, the probability that the random assignment $X$ satisfies a constraint $(v_1,v_2,\pi)$ is bounded from below by the probability that both $R(v_1)= R(v_2)$ and the partial assignment $f'_{\omega(R(v_1))}=f'_{\omega(R(v_2))}$ assigns consistent values to $v_1,v_2$.
The probability of this event is equal to the probability (over $\omega$) that the partial assignment $f'_{\omega}$ satisfies the constraint $(v_1,v_2,\pi)$ conditioned on the event that $f'_\omega$ assigns a label to either $v_1$ or $v_2$. Therefore,
letting $x_\omega\in \Sigma^W$ denote any assignment consistent with the partial assignment
  $f'_{\omega}$,
  \begin{multline}
    \Prob[X]{\pi(X_{v_1},X_{v_2})=1}
    \ge \Prob[X]{\pi(X_{v_1},X_{v_2})=1 \land R(v_1)=R(v_2)}\\
    = \Prob[\omega\sim \Omega']{\pi(x_{\omega,v_1},x_{\omega,v_2})=1 \mid f_\omega(v_1,x_{\omega,v_1})=1 \lor f_\omega(v_2,x_{\omega,v_2})=1}\\
    = \frac{\E_{\omega\sim\Omega'}
      \pi(x_{\omega,v_1},x_{\omega,v_2})
      \min\set{f'_{\omega}(v_1,x_{\omega,v_1}),f'_{\omega}(v_2,x_{\omega,v_2})}}
    {\E_{\omega\sim\Omega'}
      \max\set{f'_{\omega}(v_1,x_{\omega,v_1}),f'_{\omega}(v_2,x_{\omega,v_2})}}\mper
\label{eq:intermediate3}
  \end{multline}
  At this point, we will use the following simple inequality (see
  \pref{lem:min-vs-max} toward the end of this subsection):
  \begin{quote}
    \emph{Let $A,B,Z$ be jointly-distributed random variables such
      that $A,B$ are nonnegative-valued and $Z$ is $0/1$-valued.
      Then, $\E Z\cdot \min\set{A,B}\ge (\tfrac{1-\gamma'}{1+\gamma'})\E
      \max\set{A,B}$ as long as $\E Z \cdot \min\set{A,B}\ge (1-\gamma')\E\tfrac12(A+B)\mper$}
  \end{quote}
  We instantiate this inequality with $A=f'_\omega(v_1,x_{\omega,v_1})$, $B=f'_{\omega}(v_2,x_{\omega,v_2})$, $Z=\pi(x_{\omega,v_1},x_{\omega,v_2})$ for $\omega\sim\Omega'$, and $\gamma'=\gamma_{v_1,v_2,\pi}$, where
  \begin{displaymath}
    \E_\omega  \pi(x_{\omega,v_1},x_{\omega,v_2})
    f'_\omega(v_1,x_{\omega,v_1})
    f'_\omega(v_2,x_{\omega,v_2})=(1-\gamma_{v_1,v_2,\pi})\lambda.
  \end{displaymath}
  The condition of the $A,B,Z$-inequality corresponds to the condition
  $\E_\omega  \pi(x_{\omega,v_1},x_{\omega,v_2})
  \min\{f'_\omega(v_1,x_{\omega,v_1}),
  f'_\omega(v_2,x_{\omega,v_2})\}\ge (1-\gamma_{v_1,v_2,\pi})\lambda$,
  because $f'$ is $0/1$-valued and $\E_{\omega} f'_\omega(v,x_{\omega,v})^2=\snorm{(T\otimes \Id )f'}=\lambda$.
  The conclusion of the $A,B,Z$-inequality shows that the right-hand
side of \pref{eq:intermediate3} is bounded from below by
$\nfrac{1-\gamma_{v_1,v_2,\pi}}{1+\gamma_{v_1,v_2,\pi}}$.

  Hence, by convexity of $\varphi'(x)=(1-x)/(1+x)$,
  \begin{displaymath}
    \E_X \val(G_{sym};X)
    =\E_{(v_1,v_2,\pi)} \Prob[X]{\pi(X_{v_1},X_{v_2})=1}
    \ge \E_{(v_1,v_2,\pi)} \varphi'(\gamma_{v_1,v_2,\pi}) \ge \varphi'\Bigparen{\E_{(v_1,v_2,\pi)}\gamma_{v_1,v_2,\pi}}\mcom
  \end{displaymath}
  where $(v_1,v_2,\pi)$ are distributed according to $G_{sym}$.
  It remains to lower-bound the expectation of $\gamma_{v_1,v_2,\pi}$ over $(v_1,v_2,\pi)\sim G_{sym}$.
  \begin{align*}
    \E_{(v_1,v_2,\pi)} 1-\gamma_{v_1,v_2,\pi}
    &= \E_{(v_1,v_2,\pi)} \tfrac{1}{\lambda}
    \E_{\omega}\pi(x_{\omega,v_1},x_{\omega,v_2}) f'_\omega(v_1,x_{\omega,v_1}) f'_\omega(v_2,x_{\omega,v_2})
    \\
    &= \tfrac{1}{\lambda} \snorm{(G\ot I_{\Omega'})f'} = 1-\gamma\mper
    \qedhere
  \end{align*}
\end{proof}

\subsubsection{Some inequalities}
\label{sec:some-inequalities}

The following lemma show that if the expected geometric average of two
random variable is close to their expected arithmetic average, then
the expected minimum of the two variables is also close to the
expected arithmetic average.
A similar lemma is also used in proofs of Cheeger's inequality.

\begin{lemma}
\label{lem:gm-vs-min}
  Let $A,B$ be jointly-distributed random variables, taking nonnegative values.
  If $\E  \sqrt{AB} = \rho \E \tfrac12(A+B)$, then
  \begin{displaymath}
    \E \min\set{A,B} \ge \varphi(\rho)\cdot  \E \tfrac12(A+B)\mcom
  \end{displaymath}
  for $\varphi(x)=1-\sqrt{1-x^2}$.
\end{lemma}

\begin{proof}
  Since $\min\set{A,B}=\tfrac12(A+B)-\tfrac12\Abs{A-B}$, it is enough
  to lower bound
  \begin{align*}
    \E \tfrac12\Abs{A-B}
&    = \E \tfrac12\Abs{A^{1/2}-B^{1/2}}\Abs{A^{1/2}-B^{1/2}}\\
&    \le \Paren{\E \tfrac12\Paren{A^{1/2}-B^{1/2}}^2\cdot
      \E \tfrac12\Paren{A^{1/2}-B^{1/2}}^2}^{1/2}\\
&    = \Bigparen{\Paren{\E \tfrac12(A+B) - \E \sqrt{AB} }
      \cdot \Paren{\E \tfrac12(A+B) +\E \sqrt{AB} } }^{1/2}
    \\
&    = \sqrt{1-\rho^2} \cdot \E \tfrac12(A+B)\mper
  \end{align*}
  The second step uses Cauchy--Schwarz.
\end{proof}

The following corollary will be useful for us to prove Cheeger's
inequality for two-player games.

\begin{corollary}
\label{cor:gm-vs-min}
  Let $A,B$ be as before. Let $Z$ be a $0/1$-valued random variable,
  jointly distributed with $A$ and $B$.
  If $\E Z \cdot \sqrt{AB}= \rho \E \tfrac12(A+B)$, then $\E Z\cdot
  \min\set{A,B}\ge \varphi(\rho)\E \tfrac12(A+B)$, for $\varphi$ as before.
\end{corollary}

\begin{proof}
  The corollary follows from the convexity of $\varphi$.
  For notational simplicity, assume $\E\tfrac12(A+B)=1$ (by scaling).
  Let $\lambda=\E Z \cdot  \tfrac12(A+B)$.
  Write $\rho$ as a convex combination of $0$ and a number $\rho'$
  such that $\rho=\rho'\cdot \lambda+0\cdot (1-\lambda)$.
  Since $\E Z \sqrt{AB}=\rho' \cdot \lambda$, \pref{lem:gm-vs-min}
  implies $\E Z\min\set{A,B}\ge \varphi(\rho')\cdot \lambda$.
  The convexity of~$\varphi$ implies $\varphi(\rho)\le \lambda\cdot
  \varphi(\rho')+(1-\lambda)\varphi(0)=\lambda\cdot \varphi(\rho')$.
\end{proof}

\begin{lemma}
\label{lem:min-vs-max}
  Let $A,B,Z$ be jointly-distributed random variables as before ($A,B$
  taking nonnegative values and $Z$ taking $0/1$ values).
  If $\E Z \cdot \min\set{A,B}= (1-\gamma)\E \tfrac12(A+B)$, then
  \begin{displaymath}
    \frac{\E Z \cdot\min\set{A,B}}
    {\E \max\set{A,B}}
    \ge \frac{1-\gamma}{1+\gamma}
    \mper
  \end{displaymath}
\end{lemma}
\begin{proof}
  For simplicity, assume $\E\tfrac12(A+B)=1$ (by scaling).
  Since $\min\set{A,B}=\tfrac12(A+B)-\tfrac12\abs{A-B}$,
  we get
  $1-\gamma\le \E \min\set{A,B} = 1-\E \tfrac12\abs{A-B}$,
  which means that $\E \tfrac12\abs{A-B}\le \gamma$.
  Since
  $\max\set{A,B}=\tfrac12(A+B)+\tfrac12\abs{A-B}$,
  it follows that
  \begin{displaymath}
    \frac{\E Z\cdot \min\set{A,B}}
    {\E \max\set{A,B}}
    = \frac{1-\gamma}{1+\E\tfrac12\abs{A-B}}
    \ge \frac{1-\gamma}{1+\gamma}
    \mper\qedhere
  \end{displaymath}
\end{proof}

\subsection{Multiplicativity of $\rval$ - proof of Lemma~\ref{lem:mult}}
%
We prove the multiplicativity of $\rval$, namely that for projection games $G,G'$,
\[ \rval(G\ot G') = \rval(G)\cdot\rval(G')\mcom
\]
using ideas that were presented in \pref{sec:basic-approach}.

In this subsection a {\em vector assignment} for $G$ with value $\rho$ is a function $f:V\times\Sigma\times\Omega \to\Rnn$ such that $\snorm{(T_v \ot \Id_\Omega)f}\le 1$ for all $v$, and $\snorm{(G \ot \Id_\Omega)f}\ge \rho$.

First note that $\rval(G\ot G') \ge \rval(G)\cdot \rval(G')$ since for any two vector assignments
$f:V\times\Sigma\to\Rnn^{\Omega}$ for $G$ and
$f':V'\times\Sigma'\to\Rnn^{\Omega'}$ for $G'$,
their tensor product
$f\ot f':V\times V'\times\Sigma\times\Sigma'\to \Rnn^{\Omega\times\Omega'}$
defined by
$f\ot f'(v,v',\beta,\beta') = f(v,\beta)\ot f(v',\beta')$ obeys
\[
\snorm{(G\ot G'\ot \Id_{\Omega\times\Omega'})f\ot f'} = \snorm{(G\ot \Id_\Omega)f}\cdot \snorm{(G'\ot \Id_{\Omega'})f'}\]
Similarly, since for each $(v,v')\in V\times V'$, $T_{v,v'} = T_v\ot T_{v'}$, so
\[
\snorm{(T_{v,v'}\ot \Id_{\Omega\times\Omega'})f\ot f'} = \snorm{(T_v\ot \Id_\Omega)f}\cdot \snorm{(T_{v'}\ot \Id_{\Omega'})f'} \le 1\cdot 1 = 1\mper \]
For the converse direction, assume for the sake of contradiction that $g:V\times V'\times\Sigma\times\Sigma'\to \Rnn^{\Omega}$ is a vector assignment for $G\ot G'$ with value $\rho > \rval(G)\cdot \rval(G')$.
Define $\Omega' = V'\times \Sigma'\times\Omega$ so that now $g$ can be viewed as a nonnegative function $g:V\times\Sigma\to \Rnn^{\Omega'}$ mapping to each $(v,\beta)$ a vector.
Using the factorization
\[
(G\ot G' \ot \Id_\Omega) =
(G\ot \Id_{V'\times\Sigma'} \ot \Id_\Omega)
(\Id_{V\times\Sigma}\ot G'\ot \Id_\Omega)
\]
and letting $h = (\Id_{V\times\Sigma}\ot G'\ot \Id_\Omega)g$ we have on one hand that
\[  \snorm {(G\ot \Id_{V'\times\Sigma'} \ot \Id_\Omega)h} = \rho\mcom
\]
and on the other hand
\[ \snorm  {(T_v\ot \Id_{V'\times\Sigma'} \ot \Id_\Omega)h}  \le \rval(G')
\]
because for each $v$, the function $\sum_\beta g(v,v',\beta,\beta)$ is a function from $V'\times \Sigma'$ to $\Rnn^{\Omega}$ that itself is a vector assignment for $G'$, thus its value is upper bounded by $\rval(G')$.
If we multiply $h$ by the scalar $\frac 1{ \rval(G')}$ we get a vector assignment for $G$, whose value is $ \rho /\rval(G')\le \rval(G)$, a contradiction.

\subsection{Proof of Main Theorem, Theorem~\ref{thm:collision-val-product}}

By \pref{thm:valplus-bound} we have
\[ \snorm{G\ot H} \le \rval(G)^2 \cdot \snorm{H} \mper \]
Let $\delta = \snorm{G}$. By \pref{thm:approximation}, $\rval(G)^2 \le \varphi(\delta) =  \frac{2\sqrt \delta}{1+\delta}$, by monotonicity of $\varphi(x) := \frac {2\sqrt x}{1+x}$ on the interval $[0,1]$.
Together,
\[ \snorm{G\ot H} \le \rval(G)^2 \cdot \snorm H \le \varphi(\snorm G)\cdot \snorm H
\]

\section{Few Repetitions --- Proof of Theorem~\ref*{thm:few-repetitions}}
\label{sec:few-reps}

\Dnote{}%

The following theorem will allow us to prove a tight bound on the
value of projection games after few parallel-repetitions
(\pref{thm:few-repetitions}).

\begin{theorem}
  \label{thm:few-reps-1}
  \label{thm:two-games}
  Let $G$ and $H$ be two projection games.
  Suppose $\lVert H\rVert^2 \le 1-\gamma$ but $\norm{G\ot H}^2\ge 1-\eta-\gamma$
  (with $\gamma$ small enough, say $\gamma\le \nfrac1{100}$).
  Then,
  \begin{displaymath}
    \lVert G\rVert^2 \ge 1-O\bigparen{\eta +\sqrt{\gamma\eta}}\mper
  \end{displaymath}
\end{theorem}

Let us first explain how the bound above improves over the bounds in
the previous sections.
In the notation of the above theorem, \pref{thm:valplus-bound}
implies that $\rval(G)^2= \frac{\lVert G\otimes H\rVert^2}{\lVert H\rVert^2}\ge 1-O(\eta)$.
Thus, $\lVert G\rVert^2\ge 1-O(\sqrt \eta)$ by \pref{thm:approximation}.
We see that this bound is worse than the above bound whenever $\gamma$ is close to $0$.

Before proving the theorem, we will show how it implies
\pref{thm:few-repetitions} (parallel-repetition bound for few
repetitions).

\begin{proof}[Proof of \pref{thm:few-repetitions}]
  Let us first reformulate \pref{thm:two-games}.
  (The above formulation reflects our proof strategy for
  \pref{thm:two-games} but for our current application, the following
  formulation is more convenient.)
  Let $G$ and $H$ be two projection games.
  Suppose $\lVert G\rVert^2\le 1-\e$ and $\lVert H\rVert^2\le 1-t \e$ for some $\e>0$ and
  $1\le t\ll 1/\e$.
  Then,  \pref{thm:two-games} shows that
  \begin{displaymath}
    \lVert G\ot H\rVert^2\le 1-\Big(t+\Omega\big(\tfrac1 t\big)\Big)\cdot \e
  \end{displaymath}
  (Here, we use that for $\gamma=t\e$ and $\eta=\Omega(1/t)\e$, we have
  $\eta+\sqrt{\eta\gamma}=\Omega(\e)$.)
  From this bound, we can prove by induction on $k$ that $\lVert  G^{\ot
    k}\rVert^2 \le 1-\cramped{\Omega(k)^{1/2}}\cdot \e$ for $k\ll 1/\e^2$.
  Concretely, let $\set{t(k)}_{k\in\N}$ be the sequence such that $\lVert G^{\ot k}\rVert^2 =
  1-t(k)\cdot \e$.
  The above formulation of \pref{thm:two-games} implies $t(k+1)\ge
  t(k)+\Omega(1/t(k))$ (as long as $t(k)\ll 1/\e$).
  Since the sequence increases monotonically, it follows that
  $t(k+1)^2\ge t(k)^2+\Omega(1)$ (by multiplying the recurrence relation with $t(k+1)$ on both
  sides).
  Hence, as long as $t(k)\ll 1/\e$, we have $t(k)^2=\Omega(k)$ and
  $t(k)=\cramped{\Omega(k)^{1/2}}$, as desired.
\end{proof}

The proof of \pref{thm:two-games} follows a similar structure as the
proofs in \pref{sec:basic-approach} and \pref{sec:approx-gen}.
If $\norm{H}^2\le 1-\gamma$ and $\snorm{G}\ge 1-\eta-\gamma$, then
following the proof of \pref{thm:valplus-bound} we can construct a
function for $G$ that satisfies certain properties.
(In particular, the function  certifies $\rval(G)^2\ge 1-O(\eta)$).
The challenge is to construct an assignment for $G$ from this function.
This construction is the content of the following two lemmas,
\pref{lem:few-round-1} and \pref{lem:few-round-2}.
\Dnote{}%

\Dnote{}%

\medskip

In the following $G$ will be the ``square'' of a projection game $G_0$ (formally, the operator $G=G_0^TG_0$).
Let $\Assign_{W,\Sigma}(\Omega)$ be the set of nonnegative functions $f\in L(W\times \Sigma\times \Omega)$ such that every slice $f_{\omega}$ is deterministic, i.e., for every $v\in W$, there exists at most one label $\beta$ with $f_{\omega}(v,\beta)>0$.
We write $f_{v,*}$ to denote the function $(T_v \otimes \Id_\Omega )f$.
\begin{lemma}
  \label{lem:few-round-1}
  Let $f\in \Assign_{W,\Sigma}(\Omega)$ be a $[0,1]$-valued function
  with $\snorm{f_{u,*}}\le 1-\gamma$ for all $u\in W$ and $\normo{f}=1$.
  Suppose every slice of $f$ is a deterministic fractional assignment.
  Suppose $\iprod{f,(G\ot I_\Omega)f}\ge 1-\eta-\gamma$ and that
  $\gamma$ is small enough (say, $\gamma<\nfrac1{100}$).
  Then, there exists a $0/1$-valued function
  $f'\in\Assign_{W,\Sigma}(\Omega')$ with
  \begin{displaymath}
    \iprod{f',(G\ot I_{\Omega'})f'}
    \ge \bigparen{1-O(\eta +\sqrt{\eta\gamma})}\snorm{f'}\mcom
  \end{displaymath}
  and $\snorm{f'_{u,*}}\ge \nfrac1 3$ for all but an $O(\eta)$ fraction of vertices in $W$.
\end{lemma}

\begin{proof}
  Since $\snorm{f_{u,*}}\le 1-\gamma$ for all $u\in W$, the condition
  $\iprod{f,(G\ot I_\Omega) f}\ge 1-\eta-\gamma$ implies that
  $\snorm{f_{u,*}}\ge 0.9$ for all but $O(\eta)$ fraction of vertices $u\in W$.

  Let $\Omega'=\Omega\times [\nfrac1{10},\nfrac{9}{10}]$ (with the
  uniform measure on the interval $[\nfrac1{10},\nfrac{9}{10}]$).
  Let $f'\from W\times \Sigma\times \Omega'\to \bits$ be such that for all $u\in W$, $\alpha\in\Sigma$, $\omega\in \Omega$, $\tau\in[\nfrac 1{10},\nfrac9{10}]$,
  \begin{displaymath}
    f'_{\omega,\tau}(u,\alpha) =
    \begin{cases}
      1 & \text{ if $f_\omega(u,\alpha)^2>\tau$\mcom}\\
      0 & \text{ otherwise.}
    \end{cases}
  \end{displaymath}
  It follows that $\snorm{f'_{u,*}}\ge 1/3$ for all but $O(\eta)$
  vertices $u\in W$ (using that $\snorm{f_{u,*}}\ge 0.9$ for all but $O(\eta)$ vertices).
  We see that $f_\omega(u,\alpha)^2-\nfrac1{10}\le \E_\tau f'_{\omega,\tau}(u,\alpha)^2\le \nfrac{10}9 \,f_\omega(u,\alpha)^2 $.
  Let $B_\omega\from W\times \Sigma\to \R$ be the $0/1$ indicator
  function of the event $\set{f(u,\alpha)^2\in [\nfrac 1{10},\nfrac9{10}]}$.
  Then,
  \begin{align}
    \E_\tau \Paren{f'_{\omega,\tau}(u,\alpha)-f'_{\omega,\tau}(u',\alpha')}^2
    &\le 2\cdot\Bigparen{f_\omega(u,\alpha)-f_\omega(u',\alpha')}^2 \notag\\
    & + 2\Bigparen{B_\omega(u,\alpha)+B_\omega(u',\alpha')}\cdot
    \Abs{f_\omega(u,\alpha)^2-f_\omega(u',\alpha')^2}\label{eq:hybrid-bound}
  \end{align}
  We can see that the left-hand side is always at most
  $2\abs{f_\omega(u,\alpha)^2-f_\omega(u',\alpha')^2}$.
  Hence, the inequality holds if $B_\omega(u,\alpha)+B_\omega(u',\alpha')\ge 1$.
  Otherwise, if $B_\omega(u,\alpha)+B_\omega(u',\alpha')=0$, then the left-hand
  side is either $0$ or $1$. In both cases, the left-hand side is
  bounded by $ 2\paren{f_\omega(u,\alpha)-f_\omega(u',\alpha')}^2$.

  For every $\omega\in\Omega$, let $x_\omega\in\Sigma^V$ be an assignment
  consistent with the fractional assignment $f_\omega$.
  Let $h_\omega\from W\to [0,1]$ be the function
  $h_\omega(u)=f_\omega(u,x_{\omega,u})$.
  Similarly, $h'_{\omega,\tau}(u)=f'_{\omega,\tau}(u,x_{\omega,u})$.
  Let $G_\omega$ be the linear operator on functions on $W$, so that for every $g\from W\to\R$
  \begin{displaymath}
    \iprod{g,G_\omega g}
    = \E_{(u,v,\pi)\sim G} \pi(x_{\omega,u},x_{\omega,v}) \cdot g(u)\cdot g(v)
    \mper
  \end{displaymath}
  (As a graph, $G_\omega$ corresponds to the constraint
  graph of $G$, but with some edges deleted.)
  Let $L_\omega$ be the corresponding Laplacian, so that for all $g\from W\to \R$
  \begin{displaymath}
    \iprod{g,L_\omega g}
    = \E_{(u,v,\pi)\sim G} \pi(x_{\omega,u},x_{\omega,v}) \cdot \tfrac12 \bigparen{g(u)- g(v)}^2
    \mper
  \end{displaymath}
  With these definitions, $\iprod{f_\omega,G
    f_\omega}=\iprod{h_\omega,G_\omega h_\omega}$ and
  $\iprod{f'_{\omega,\tau}, G f'_{\omega,\tau}}=\iprod{h'_{\omega,\tau},G_\omega h'_{\omega,\tau}}$.
  We also use the short-hand, $B_\omega(u)=B_\omega(u,x_{\omega,u})$.
  With this setup, we can relate
  $\E_{\tau}\iprod{h'_{\omega,\tau},L_\omega,h'_{\omega,\tau}}$ and
  $\iprod{h_\omega,L_\omega,h_\omega}$,
  \begin{align*}
    &\E_{\tau} \iprod{h'_{\omega,\tau},L_\omega h'_{\omega,\tau}} \\
    &\le 2     \iprod{h_{\omega},L_\omega h_{\omega}}
    +  \E_{(u,v,\pi)\sim G} \bigparen{B_\omega(u)+B_\omega(v)}
    \cdot \pi(x_{\omega,u},x_{\omega,v})\cdot \Abs{h_\omega(u)^2-h_\omega(v)^2}
    \quad\by{\pref{eq:hybrid-bound}}\\
    & \le 2     \iprod{h_{\omega},L_\omega,h_{\omega}}
    + 10 \iprod{h_\omega,L_\omega h_\omega}^{1/2} \cdot \norm{B_\omega}
    \quad\using{Cauchy--Schwarz}\mper
  \end{align*}
  Let $M_\omega$ be the linear operator on functions on $W$ with the following quadratic form,
  \begin{displaymath}
    \iprod{h_\omega,M_\omega h_\omega} = \E_{(u,v,\pi)\sim G} (1-\pi(x_u,x_v))
    \tfrac12\bigparen{h_\omega(u)^2+h_\omega(v)^2}
  \end{displaymath}
  Let $L'_\omega=L_\omega+M_\omega$.
  The following identity among these operators holds
  \begin{equation}
    \label{eq:laplacian-identity}
    \iprod{h_\omega,G_\omega h_\omega} = \snorm{h_\omega} - \iprod{h_\omega,L'_\omega h_\omega}
    = \snorm{h_\omega} - \iprod{h_\omega,L_\omega h_\omega} - \iprod{h_\omega,M_\omega h_\omega}.
  \end{equation}
  Let $\gamma_\omega=\normo{h_\omega}-\snorm{h_\omega}$,
  and let $\eta_\omega=\iprod{h_\omega ,L'_\omega  h_\omega}$.
  We see that $\E_{\tau}\iprod{h'_{\omega,\tau},M_\omega
    h'_{\omega,\tau}}\le 2 \iprod{h_\omega,M_\omega h_\omega}$.
  Thus, $\E_\tau \iprod{h'_{\omega,\tau},L'_\omega
    h'_{\omega,\tau}}\le O(\eta_\omega + \sqrt{\eta_\omega}\cdot
  \norm{B_\omega})$ (using \pref{eq:hybrid-bound}).

  Next, we claim that $\snorm{B_\omega} = O(\gamma_\omega)$.
  On the whole domain $W$, we have $B_\omega\le 100 (1-h_\omega)h_\omega=100 (h_\omega -h_\omega^2)$.
  So, $\snorm{B_\omega} \le 100 (\normo{h_\omega}-\snorm{h_\omega})=100\gamma_\omega$.
  Hence,
  \begin{equation}
    \label{eq:laplacian-bound}
    \E_\tau \iprod{h'_{\omega,\tau},L'_\omega
      h'_{\omega,\tau}}\le O(\eta_\omega +
    \sqrt{\eta_\omega\gamma_\omega}).
  \end{equation}
  Using \pref{eq:laplacian-identity} and the relation between $G$ and
  $G_\omega$, we can integrate \pref{eq:laplacian-bound} over $\Omega$,
  \begin{align*}
    \iprod{f',(G\ot I_{\Omega'})f'}
    =\int_\Omega \E_{\tau}\iprod{h'_{\omega,\tau},G h'_{\omega,\tau}} \, d\omega
    &= \snorm{f'} - O(1)\int_\Omega \eta_\omega + \sqrt{\eta_\omega\gamma_\omega} \, d\omega\\
    &= \snorm{f'} - O(\eta + \sqrt{\eta\gamma})\mper
  \end{align*}
  The last step uses that $\int_\omega \eta_\omega =
  \snorm{f}-\iprod{f,(G\ot I_\Omega)f}\le \eta$ and $\int
  \gamma_\omega = \normo{f}-\snorm{f}\le \gamma$, as well as
  Cauchy--Schwarz.
  Since $\snorm{f'}\ge 0.9-\eta-\gamma$ (using that $\norm{f}\ge
  1-\eta-\gamma$), we see that $\iprod{f',(G\ot I_{\Omega'})f'}\ge (1-O(\eta+\sqrt{\eta\gamma}))\snorm{f'}$.
\end{proof}

\begin{lemma}
\label{lem:few-round-2}
  Let $f\in\Assign_{W,\Sigma}(\Omega)$ be a $0/1$-valued function with
  $\norm{f_{u,*}}\ge \nfrac13$ for all but $O(\e)$ vertices, and
  $\iprod{f,(G\ot I_\Omega)f}= (1-\e)\snorm{f}$.
  Then,
  \begin{displaymath}
    \val(G)\ge 1-O(\e)\mper
  \end{displaymath}
\end{lemma}

\begin{proof}
  Rescale $\Omega$ so that it becomes a probability measure.
  Let $\lambda$ be the scaling factor so that $\nfrac13 \lambda\le \snorm{f_{u,*}}\le
  \lambda$ for all $u\in W$ %
  (after
  rescaling of $\Omega$).
  Following the analysis in \pref{lem:correlated-sampling}, there
  exists an assignment for $G$ with value at least $\val(G)\ge
  (1-\e')/(1+\e')$ for
  \begin{displaymath}
    \e'=
    \E_{(u,v,\pi)\sim G}
    1-\tfrac{2}{\snorm{f_{u,*}}+\snorm{f_{v,*}}} \cdot \E_{\omega}
    \pi(x_{\omega,u},x_{\omega,v})\cdot f_{\omega}(u,x_{\omega,u}) f_{\omega}(u,x_{\omega,v})\mper
  \end{displaymath}
  Here, $x_\omega\in \Sigma^W$ is a consistent assignment with
  $f_\omega$ for all $\omega\in \Omega$.
  We are to show $\e'=O(\e)$.
  To ease notation, let use define two jointly-distributed random
  variables,   for $(u,v,\pi)\sim G$,
  \begin{align*}
    \Delta\cdot \snorm{f}
    &=\tfrac12(\snorm{f_{u,*}}+\snorm{f_{v,*}})
    \\
    \Gamma\cdot \snorm{f}
    &=\tfrac12(\snorm{f_{u,*}}+\snorm{f_{v,*}})-\E_{\omega}
    \pi(x_{\omega,u},x_{\omega,v})\cdot f_{\omega}(u,x_{\omega,u}) f_{\omega}(u,x_{\omega,v})
  \end{align*}
  In expectation, %
  $\E \Gamma=\e$.
  Also, %
  $0\le\Gamma\le \Delta$  with probability $1$.
  Furthermore, $\Prob{\Delta<\nfrac13}=O(\e)$.
  We can express $\e'$ in terms of these variables,
  \begin{displaymath}
    \e' = \E 1 - \tfrac{1}{\Delta} \cdot (\Delta-\Gamma) = \E \tfrac{1}{\Delta}\Gamma\mper
  \end{displaymath}
  Let $B$ (for bad) be the $0/1$-indicator variable of the event
  $\set{\Delta<-\nfrac23}$ (which is the same as $\set{1/(1+\Delta)>3}$).
  Since this event happens with probability at most $O(\e)$, we have $\E B=O(\e)$.
  Then, we can bound $\e'$ as
  \begin{displaymath}
    \e' = \E \tfrac{1}{\Delta}\Gamma
    = \E B \cdot \underbrace{\tfrac{1}{\Delta}\Gamma}_{\le 1} \,+\, \underbrace{(1-B)\cdot \tfrac{1}{\Delta}}_{\le 3}\cdot\,\Gamma
    \le \E B + 3 \E \Gamma = O(\e)\mper
    \qedhere
  \end{displaymath}
\end{proof}

\begin{proof}[Proof of \pref{thm:two-games}]
  Let $G$ and $H$ be two projection games.
  Suppose $G$ has vertex sets $U$ and $V$ and alphabet $\Sigma$.
  Suppose $\lVert H\rVert^2 \le 1-\gamma$ but $\norm{G\ot H}^2\ge 1-\eta-\gamma$
  (with $\gamma$ small enough, say $\gamma\le \nfrac1{100}$).
  We are to show
  \begin{math}
    \lVert G\rVert^2 \ge 1-O\paren{\eta +\sqrt{\gamma\eta}}\mper
  \end{math}

  Let $f$ be an optimal assignment for $G\otimes H$ so that $\snorm{(G\otimes H) f}\ge 1-\eta-\gamma$.
  Consider the function $h=(\Id\otimes H) f$.
  This function satisfies $\snorm{(G\otimes \Id) h}\ge 1-\eta -\gamma$.
  Since $\lVert  H \rVert^2 \le 1-\gamma $, we know that $\snorm{(T_v \otimes \Id) h}\le 1-\gamma$ for all vertices $v\in V$.
  As before (by a convexity argument), we can derandomize $h$ so that $h\in \Assign_{V,\Sigma}(\Omega)$.
  Since $0\le h\le 1$, we can apply \pref{lem:few-round-1} to the operator $G'=G^TG$ and the function $h$ to obtain a $0/1$-valued function $h'\in \Assign_{V,\Sigma}(\Omega')$ such that $\norm{(G\otimes \Id) h'}\ge 1-O(\eta +\sqrt{\eta\gamma})\snorm{h'}$ and $\snorm{h'_{v,*}}\ge 1/3$ for all but an $O(\eta)$ fraction of vertices in $V$.
  (The transposition in $G^T$ is with respect to the underlying inner product so that $\iprod{p,G q}=\iprod{G^Tp,q}$ for all functions $p$ and $q$.)
  By \pref{lem:few-round-2}, we can use this function to get an assignment for $G'$ of value at least $1-O(\eta + \sqrt{\eta\gamma})$.
  The value of an assignment in $G'$ corresponds exactly to the collision value of the assignment in $G$.
  Therefore,  $\snorm{G}\ge 1-O(\eta+\sqrt{\gamma \eta})$.
\end{proof}

\section{Inapproximability Results}\label{sec:inapprox}
\subsection{\labelcover}

In this section we prove a new hardness result for \labelcover (\pref{thm:labelcover-new} below), and derive \pref{cor:labelcover}.

\begin{definition}[\labelcover]
Let $\eps:\N\to [0,1]$ and let $s:\N\to\N$ be functions. We define the $\labelcover_s(1,\eps)$ problem to be the problem of deciding if an instance of \labelcover of size $n$ and alphabet size at most $s(n)$ has value $1$ or at most $\eps(n)$.
\end{definition}
When we refer to a reduction from 3SAT to $\labelcover_s(1,\eps)$, we mean that satisfiable 3SAT instances are mapped to \labelcover instances with value $1$, and unsatisfiable 3SAT instances are mapped to \labelcover instances with value at most $\eps$.

Let us begin by reviewing the known results. We mentioned in \pref{sec:labelcover-short} that the PCP theorem~\cite{AS,ALMSS} implies that $\labelcover_s(1,1-\delta)$ for some constant $s$ and (small enough) $\delta>0$. Parallel repetition applied to this instance $k$ times implies, via Raz's bound that $\labelcover_{a^k}(1,\beta^{k})$ is NP-hard, for some $\beta<1$ and some constant $a>1$. So, taking $k = O(\log 1/\eps)$ will imply a soundness of $\eps$, with an alphabet of size $s = \poly(1/\eps) = a^k$. (This is proved in \pref{sec:labelcover-short}).

\begin{theorem}[PCP theorem followed by Raz's theorem]\label{thm:PCPandRaz}
There are some absolute constants $a>1$ and $0<\beta<1$ such that for every $k\in \N$ there is a reduction that takes instances of 3SAT of size $n$ to instances of $\labelcover_s(1,\eps)$ of size $n^{O(k)}$, such that $s \le a^k$ and $\eps \le \beta^k$, and in particular, $s \le \poly(1/\eps)$.
\end{theorem}

In fact, one can take $k = \omega(1)$ in the above and get a reduction from the original \labelcover to $\labelcover_s(1,\eps)$ still with $s = \poly(1/\eps) = a^k$ but now the size of the instance grows to be $n^k$. Setting $k=\log\log n$ or $\log n$ is still often considered reasonable, and yields quasi NP hardness results, namely, hardness results under quasi-polynomial time reductions.

For strictly polynomial-time reductions, the work of Moshkovitz and Raz \cite{MR10} gives hardness for \labelcover with sub-constant value of $\eps$,
\begin{theorem}[Theorem 10 in \cite{MR10}]\label{thm:MR}
There exists some constant $c>0$ such that the following holds. For every $\eps:\N\to[0,1]$ there is a reduction taking 3SAT instances of size $n$ to $\labelcover_s(1,\eps)$ instances of size $n^{1+o(1)}\cdot\poly(1/\eps)$ such that $s \le \exp(1/\eps^c)$.
\end{theorem}

By applying parallel repetition to an instance of \labelcover from the above theorem, and using the bound of \pref{cor:parrep-small-soundness}, we get,

\begin{theorem}[New NP-hardness for \labelcover]\label{thm:labelcover-new}
For every constant $\alpha>0$ the following holds. For every $\eps:\N\to[0,1]$ there is a reduction taking 3SAT instances of size $n$ to $\labelcover_s(1,\eps)$ instances of size $n^{O(1)}\cdot\poly(1/\eps)$ such that $s \le \exp(1/\eps^\alpha)$.
\end{theorem}
The improvement of this theorem compared to \pref{thm:MR} is in the for-all quantifier over $\alpha$.
\begin{proof}
Assume $\eps = o(1)$ otherwise \pref{thm:PCPandRaz} can be applied with $k = O(\log 1/\eps)$ resulting in a much smaller bound on $s$. The reduction is as follows. Starting with a 3SAT instance of size $n$, let $G$ be the $\labelcover_{s_1}(1,\eps_1)$ instance output of the reduction from \pref{thm:MR} with
$\eps_1 = (3c/\alpha \cdot \eps^{\alpha} )^{1/c}$, and output $G^\ok$ for $k = 3c/\alpha$.

The resulting instance $G^\ok$ has size $n^{O(1)}$, alphabet size $s =(s_1)^k$ and the soundness is at most $(4\eps_1)^{k/2} \le (4\cdot (\frac {3c}\alpha)^{1/c} \eps^{\alpha/c})^{3c/2\alpha} \le \eps$ by \pref{cor:parrep-small-soundness} (and assuming $\eps=o(1)$). Finally, plugging in the bound for $s_1$ and the value of $\eps_1$,
\[
s = (s_1)^k \le \exp( k / \eps_1^c ) = \exp( 1 / \eps^\alpha )\mper
\]
\end{proof}

Finally, \pref{cor:labelcover} follows immediately from the above theorem, by taking $\eps = (\log n)^{-c}$ and choosing $\alpha <1/c$ so that the alphabet size is bounded by $s \le \exp(1/\eps^\alpha) \le n$. We also remark that the resulting instance is regular, due to the regularity of the \cite{MR10} instance.

\subsection{\setcover}
In this section we prove \pref{cor:setcover}. First, a brief background. Feige~\cite{Feige-setcover} proved (extending \cite{LundY94}) that \setcover is hard to approximate to within factor $(1-o(1))\ln n$, by showing a quasi-polynomial time reduction from 3SAT. His reduction has two components: a multi-prover verification protocol, and a set-partitioning gadget. Moshkovitz~\cite{Moshkovitz12} shows that the multi-prover protocol can be replaced by a \labelcover instance with an {\em agreement soundness} property that she defines. She also shows how to obtain such an instance starting with a standard \labelcover instance.
\pref{cor:setcover} follows by instantiating this chain of reductions starting with a \labelcover instance from \pref{thm:labelcover-new}. Details follow.

Let $\alpha>0$ and let $G$ be an instance of \labelcover as per \pref{cor:labelcover}, with soundness $\eps < a^2 /(\log n)^4$ for $a=O(\alpha^5)$, and with $\card\Sigma\le n$. This clearly follows by setting $c=5$ in the corollary, and we can also assume that $G$ is regular.
Using a reduction of Moshkovitz (Lemma 2.2 in \cite{Moshkovitz12}) we construct in polynomial-time a \labelcover instance $G_1$ such that the Alice degree of $G_1$ is $D = \Theta(1/\alpha)$, and such that, setting $\eps_1 = a/(\log n)^2$,
\begin{itemize}
\item If $\val(G)=1$ then $\val(G_1)=1$
\item If $\val(G) < \eps_1^2 = a^2/ (\log n)^4 $ then every assignment for $G_1$ has the following {\em agreement soundness} property\footnote{In \cite{Moshkovitz12} the projections go from Alice to Bob, while ours go from Bob to Alice.}. For at least $1-\eps_1$ fraction of the vertices $u\in U_1$ the neighbors of $u$ project to $D$ distinct values (i.e., they completely disagree).
\end{itemize}
\inote{}
While a hardness for \setcover is proven in \cite{Moshkovitz12}, it is unfortunately proven under a stronger conjecture than our  \pref{cor:labelcover}. For completeness, we repeat the argument. Let us recall the gadget used in \cite{Feige-setcover} (see Definition 3.1 in \cite{Feige-setcover})  \begin{definition}[Partition Systems] A partition system $B(m,L,k,d)$ has the following properties:
\begin{itemize}
\item $m$: There is a ground set $[m]$. (We will use $m = n_1^D$, where $n_1$ is the size of $G_1$).
\item $L$: There is a collection of $L$ distinct partitions, $p_1,\ldots,p_L$. (We will use $L = \card{\Sigma} \le n$).
\item $k$: For $1\le i \le L$, partition $p_i$ is a collection of $k$ disjoint subsets of $[m]$ whose union is $[m]$. (We will use $k=D = O(1/\alpha)$).

    Let us denote $p_i(j)$ the $j$-th set in the partition $p_i$.
\item $d$: Any cover of $[m]$ by subsets that appear in pairwise different partitions requires at least $d$ subsets. (We will use $d = k\cdot (1-\frac 2D)\ln m$).
\end{itemize}
\end{definition}
Such a gadget is explicitly constructed in time linear in $m$ and termed ``anti-universal sets'' in \cite{NaorSS95}).
The intention is that the gadget can be covered in the yes case by $k$ sets, and in the no case by at least $d = k\cdot (1-\frac 2D)\ln m$ sets.

\paragraph{Construction of \setcover instance}
Finally, the \setcover instance will have a ground set $U_1\times [m]$ consisting of $\card{U_1}$ partition-system gadgets. For every $v\in V_1$ and value $\beta\in \Sigma$ there will be a set $S_{v,\beta}$ in the \setcover instance. Denoting by $D_V$ the number of neighbors of $v$ in $G_1$, the set $S_{v,\beta}$ will be the union of $D_V$ sets, one for each neighbor $u$ of $v$. We arbitrarily enumerate the neighbors of each $u\in U_1$ with numbers from $1$ to $D$, and then denote by $j_{uv}\in [D]$ the number on the edge from $u$ to $v$. With this notation, $S_{v,\beta}$ will be the union of the sets $\{u\}\times p_\alpha(j)$ where $j$ is the number of the edge $uv$ and $\alpha = \pi_{uv}(\beta)$:
 \[
S_{v,\beta} = \bigcup_{u\sim v} \set{u}\times (p_{\pi_{uv}(\beta)}(j_{uv})) \mper
 \]

\paragraph{Completeness} It is easy to see that if $f,g$ is a satisfying assignment for $G_1$ then by taking the sets $S_{v,f(v)}$ the gadget corresponding to $u$ is covered by the $k$ sets in the partition corresponding to $\alpha=g(u)$. In total the set cover has size $\card{V_1} = D\card{U_1}$.

\paragraph{Soundness} In the no case, we claim that every set cover has size at least $Z = (1-\frac 4D)\ln m \cdot \card{V_1}$. Assume otherwise, and let $s_u$ be the number of sets in the cover that touch $\set u \times [m]$. By \[ \sum_u s_u = Z \cdot D_V \mcom\] at least $\frac 2D$ fraction of the vertices $u\in U_1$ have $s_u < \ell \defeq (1-\frac 2D)D\ln m$, and we call such $u$'s {\em good}. Define  a randomized assignment $f:V\to\Sigma$ by selecting for each $v$ a value $\beta$ at random from the set of $\beta$'s for which $S_{v,\beta}$ belongs to the set cover (or, if no such set exists output a random $\beta$).

First, by the property of the gadget, since the set cover must cover each $\set{u}\times[m]$, if $u$ is good there must be two neighbors $v_1,v_2 \sim u$ such that the sets $S_{v_1,\beta_1}$ and $S_{v_2,\beta_2}$ are in the cover, and such that $\pi_{uv_1}(\beta_1) = \alpha = \pi_{uv_2}(\beta_2)$ (both $p_{\alpha}(j_{uv})\subset S_{v_1,\beta_1}$ and $p_{\alpha}(j_{uv_2})\subset S_{v_2,\beta_2}$).

Next, observe that if $u$ is good then each neighbor $v$ of $u$ has at most $\ell =(1-\frac 2D)D\ln m$ values of $\beta$ for which $S_{v,\beta}$ is in the cover, because each such set is counted in $s_u$.

So the probability that $f(v_1)=\beta_1$ and $f(v_2)=\beta_2$ is at least $\frac 2D \cdot \frac 1{\ell^2} >\eps_1$ and this contradicts the agreement soundness property of $G_1$ as long as $\eps_1 < a/(\log n)^2$, for $a<O(1/D^5) = O(\alpha^5)$.

\paragraph{Size}
The size of the setcover instance is at most $m=n_1^D$ times $n_1$, so $\ln (n_1^{D+1}) = (D+1)\ln n_1 = (1+\frac 1 D)\ln m$. So by choosing the appropriate constant relation between $1/D$ and $\alpha$ we get a factor of $(1-\alpha)\ln N$ hardness for approximating a \setcover instance of size $N$.

\section{Conclusions}

In many contexts, tight parallel repetition bounds are still open, for
example, general\footnote{ An earlier version of this manuscript
  erroneously claimed a reduction from general constraints to
  projection constraints.
  We are thankful to Ran Raz for pointing out this error.
  } (non-projection) two-player games,
XOR games, and games with more than two parties.

It is an interesting question whether the analytical approach in this
work can give improved parallel repetition bounds for these cases.
Recently, the authors together with Vidick gave a bound on entangled games (where the two players share entanglement), following the approach of this work.

A more open-ended question is whether analytical approaches can
complement or replace information-theoretic approaches in other
contexts (for example, in communication complexity) leading to new
bounds or simpler proofs.

\subsection*{Acknowledgments}

We thank Boaz Barak, Ryan O'Donnell, Ran Raz, and Oded Regev for
insightful discussions and comments about this work.

\addreferencesection
\newcommand{\etalchar}[1]{$^{#1}$}
\providecommand{\bysame}{\leavevmode\hbox to3em{\hrulefill}\thinspace}
\providecommand{\MR}{\relax\ifhmode\unskip\space\fi MR }
\providecommand{\MRhref}[2]{%
  \href{http://www.ams.org/mathscinet-getitem?mr=#1}{#2}
}
\providecommand{\href}[2]{#2}

\newpage

\appendix

\section{Additional Proofs}

\paragraph{Reduction to expanding games}\label{sec:expanding-reduction}
\begin{claim}\label{claim:expanding}
There is a constant $\gamma>0$ and a simple transformation from a projection game $G$ to a projection game $G'$ such that $\val(G') = \frac 1 2 + \frac {\val(G)}2$ and such that $G'$ has eigenvalue gap $\ge \gamma$ . If $G$ is $(c,d)$-regular (i.e. the degree of any $u\in U$ is $c$ and the degree of any $v\in V$ is $d$) then $G'$ can be made $(c,2d)$-regular.
\end{claim}
\begin{proof}
Let $U,V$ be the vertices of $G$, and let $\mu$ be the associated distribution of $G$. Define $U' = U\cup \set{u_0}$, and a new distribution $\mu'$ that is with probability $1/2$ equal to $\mu$ and with probability $1/2$ equal to the product distribution given by $\set{u_0}\times \mu_V\times \set{\pi_0}$ where $\pi_0$ is a 'trivial' constraint that accepts when Alice answers $1$ and regardless of Bob's answer. The claim about the value is easily verified. As for the eigenvalue gap, it follows because the symmetrized game corresponding to $\set{u_0}\times \mu_V\times \set{\pi_0}$ is the complete graph, with eigenvalue gap $1$.

If $G$ is $(c,d)$ regular and we want $G'$ to also be $(c,2d)$-regular, then instead of adding one vertex $u_0$ we add a set $U_0$ of $\card U$ new vertices, so $U' = U\cup U_0$, and place a $(c,d)$-regular graph $G_0$ between $U_0$ and $V$, accompanied with $\pi_0$ constraints as above. We are free to choose the structure of $G_0$, and if we take it to be such that its symmetrized graph has eigenvalue gap $2\gamma$, then the eigenvalue gap of $G'$ will be at least $\gamma$. Clearly the distribution $\mu'$ is uniform on $U'$ which means that with probability $1/2$ it chooses a trivial constraint, and with probability $1/2$ a $G$ constraint. Hence, $\val(G') = \frac 1 2 + \frac 1 2 \val (G)$.
\end{proof}

\section{Feige's Game}\label{app:Feige}

Uri Feige~\cite{Feige91} describes a game $G$ (which he calls NA for non-interactive agreement) for which $\val(G) = \val (G\otimes G) = \frac 12$.
This example rules out results of the form $\val(G\otimes H) \le val'(G)\cdot \val(H)$, for a useful approximation $val'$ of $\val$ thus perhaps discouraging the search for game-relaxations that are multiplicative. Our \pref{thm:collision-val-product} sidesteps this limitation by proving a multiplicative relation for $\norm G$ and not for $\val(G)$. We prove, in \pref{thm:collision-val-product} that
$\norm{G\ot H}\le \rval(G)\cdot\norm{H}$ and that the relaxation $\rval(G)$ is bounded away from $1$ for any game $G$ whose value is bounded away from $1$.

Our theorem implies that there is no projection game (including Feige's NA game) for which $\norm{G\ot G}=\norm G$. Indeed we calculate below and show that for Feige's game, $\norm{G\ot G} = \snorm G  = 1/2$.

Feige's game is defined over question set $U = \bits = V$ where each pair of questions $(u,v)\in \bits^2$ is equally likely (this is called a free game). The players win if they agree on a player and her question, namely, the alphabet is $\Sigma = \set{A0,A1,B0,B1}$ and for each question pair $(u,v) \in \bits^2$, there are two acceptable pairs of answers: $(Au,Au)$ and $(Bv,Bv)$ (where $u,v$ should be replaced by the bit value $0$ or $1$).

Is this a projection game? It is true that every answer for Bob leaves {\em at most} one possible accepting answer for Alice, but note that some answers leave no possible accepting answer. Nevertheless, the inequality between the value and the collision value established in \pref{claim:collision-val} holds for such games.
\begin{claim}
$\gnorm{G} = \half$.
\end{claim}
\begin{proof}
It is easy to see that $\gnorm{G} = \frac 1 2$, by having Bob always reply with $A0$, i.e., setting Bob's strategy to be $f(0,A0)=1, f(1,A0)=1$, and then
 \[ \snorm{Gf} = \half \sum_{\alpha\in \Sigma} Gf(0,\alpha)^2 + \half \sum_{\alpha\in \Sigma} Gf(1,\alpha)^2 = \half \cdot 1 + \half \cdot 0 = \half \mper
 \]
\end{proof}
Feige's strategy achieving $\val(G\otimes G)=\half$ is one where Bob hopes that his second question is equal to Alice's first question, and thus answers a question $v_1v_2$ by $(Av_2,Bv_2)\in \Sigma^2$.
\begin{claim}
Let $f$ be Feige's strategy attaining $\val(G\otimes G)=\half$ (defined as $f(v_1v_2,Av_2Bv_2)=1$ for each $v_1v_2$). Then $\snorm{G\otimes Gf}=\frac 1 4$.
\end{claim}
\begin{proof}
Note that $\val(G\otimes G,f) = \half$ because with probability $1/2$ $v_2=u_1$ and if Alice takes a symmetrical strategy $g(u_1u_2,Au_1Bu_1)=1$ then the players win. Let us analyze $\snorm{G\otimes G f}$ for this $f$. $Gf$ looks similar for all question tuples $u_1u_2 \in \set{00,01,10,11}$ that Alice might see. For example, suppose $u_1u_2 = 00$. Of the four possible tuples $v_1v_2$ that Bob may have received, only $v_1v_2=00$ and $v_1v_2=10$ will contribute to $Gf(00,A0B0)$. Indeed
 \[ Gf(00,A0B0) = \frac { f(00,A0B0) + f(01, A0B0) + f(10, A0B0)+ f(11, A0B0)}4 = \frac {1 +0 + 1 +0 }4 = \half\]
For any other $\alpha \in \Sigma^2$, $Gf(00,\alpha)=0$, because the answer $A1B1$ is not acceptable when Alice's questions are $u_1u_2=00$. Thus, we get $\snorm{G\otimes Gf} = \frac 1 4 \cdot 4 \cdot (\frac 1 4 )^2  = \frac 1 4$.
\end{proof}
One would guess also that $\gnorm{G\otimes G} = \frac 1 4$ but we have not analyzed this. \inote{}
\inote{}

\end{document}